\documentclass{IEEEtran}
\usepackage{cite}
\usepackage{bm}
\usepackage{amsthm}
\usepackage{amsmath,amssymb,amsfonts}
\usepackage{algorithmic}
\usepackage{setspace}
\usepackage{graphicx}
\usepackage{textcomp}
\newtheorem{Thm}{Theorem}
\newtheorem{Lem}{Lemma}
\newtheorem{Prop}{Proposition}
\newtheorem{Def}{Definition}
\newtheorem{Asum}{Assumption}
\newtheorem{Rmk}{Remark}

\def\BibTeX{{\rm B\kern-.05em{\sc i\kern-.025em b}\kern-.08em
    T\kern-.1667em\lower.7ex\hbox{E}\kern-.125emX}}
\begin{document}
\title{A Tunable Universal Formula for Safety-Critical Control}
\author{Ming Li, Zhiyong Sun, Patrick J. W. Koelewijn, and Siep Weiland
\thanks{Z. Sun was
supported in part by ``The Fundamental Research Funds for the Central Universities, Peking University''.   \emph{(Corresponding author: Zhiyong~Sun.)}}
\thanks{The authors are with the Department of Electrical Engineering, Eindhoven University of Technology, and also with the Eindhoven Artificial Intelligence Systems Institute, PO Box 513, Eindhoven 5600 MB, The Netherlands. Z. Sun is with College of Engineering, Peking University, Beijing 100871, China
{\tt\small \{ m.li3, p.j.w.koelewijn, s.weiland\}@tue.nl, zhiyong.sun@pku.edu.cn}
}}

\maketitle

\begin{abstract}
Sontag's universal formula is a widely used technique for stabilizing control through control Lyapunov functions. Recently, it has been extended to address safety-critical control by incorporating control barrier functions (CBFs). However, deriving a universal formula that satisfies requirements on essential properties, including safety, smoothness, and robustness against input disturbances, is still an open problem. To address this challenge, this paper introduces a novel solution – a tunable universal formula – by incorporating a (state-dependent) tunable term into Sontag's formula. This tunable term enables the regulation of safety-critical control performances, allowing the attainment of desired properties through a proper selection of tunable terms. Generally, the tunable universal formula can be seen as a controller that improves the quadratic program (QP)-synthesized controllers in terms of robustness and smoothness, while also reducing the conservatism (corresponding to robustness) in Sontag’s formula. Furthermore, we extend the tunable universal formula to address safety-critical control problems with norm-bounded input constraints, showcasing its applicability across diverse control scenarios. Finally, we demonstrate the efficacy of our method through a two-link manipulator safe tracking example, investigating the essential properties including safety, smoothness, and robustness against input disturbances under various tunable terms.
\end{abstract}


\section{Introduction}
\label{sec:introduction}
\IEEEPARstart{S}{afety}-critical control refers to the design of control systems within environments characterized by strict safety constraints~\cite{Zeroing_CBF}. These constraints involve physical hardware limits (e.g., workspace, joint position, and velocity constraints in robotic arms~\cite{robotic_arm}) and controller constraints for safe system operations (e.g., collision, contact force, and range constraints in quadrotor applications~\cite{wu2016safety}). Generally, designing controllers to incorporate these safety-critical constraints into practical applications is challenging, which requires achieving the forward invariance of a safe set, defined as super-level sets of scalar constraint functions~\cite{CBF_Tutorial}.
%

Control barrier functions (CBFs), which utilize Lyapunov-like arguments to ensure set forward invariance, have received significant attention in recent years~\cite{Zeroing_CBF,robotic_arm,wu2016safety,CBF_Tutorial,cohen2023characterizing,ISSf,morris2013sufficient,Safety_Margin,Inverse_Optimality}. Generally, CBFs are advantageous for handling nonlinear systems, enabling real-time control, and effectively managing high-relative-degree constraints~\cite{xiao2021high}. These advantages have led to their application in various safety-critical scenarios~\cite{Zeroing_CBF}, such as adaptive cruise control, lane keeping, and bipedal robot walking. Within existing studies, a prevalent application of CBFs is the formulation of quadratic programs (QPs) for controller synthesis~\cite{Zeroing_CBF}. The fundamental properties of the QP-synthesized controller, such as robustness against input disturbances, smoothness (or Lipschitz continuity), and inverse optimality, are studied in~\cite{cohen2023characterizing,ISSf,morris2013sufficient,Safety_Margin,Inverse_Optimality}. Among the various properties, smoothness and robustness are two crucial properties in safety-critical control design. Specifically, smooth control laws are essential for implementing certain safety-critical control methods (e.g. safe backstepping algorithms through multi-layer cascaded dynamics~\cite{cohen2024safety}) and their preference in practical applications like ensuring passenger comfort in autonomous driving~\cite{xiao2023barriernet}. On the other hand, robustness is crucial because input disturbances in a dynamical model often degrade system performance and compromise safety guarantees in real-world implementations~\cite{Safety_Margin,ISSf}. However, it was revealed that the QP-synthesized controller may exhibit non-smooth behavior~\cite{cohen2023characterizing} and requires additional modifications to enhance its robustness against input disturbances~\cite{ISSf}. As an alternative, Sontag's universal formula~\cite{Sontag_fromula}, which is a frequently used tool in control Lyapunov theory for the design of stabilizing controllers, has been adapted to address challenges in safety-critical control using CBFs~\cite{wieland2007constructive,li2023graphical,cohen2023characterizing}. In contrast to the QP-synthesized controller, Sontag's formula provides a smooth controller and additional benefits, including enhanced robustness against input disturbances in safety-critical control designs~\cite{mestres2023feasibility}.
However, using Sontag's formula in safety-critical control lacks flexibility. This often results in a controller that is ``overly robust'' in terms of safety guarantees, corresponding to the concept of safety conservatism~\cite{alan2021safe}, which can lead to excessively conservative behavior~\cite{cohen2023characterizing}. Therefore, it is sometimes necessary to reduce the level of robustness to achieve desirable performance.

According to the above discussions, we conclude that both the QP-synthesized controller and Sontag's formula may fail to achieve the desired safety, smoothness, and robustness in certain applications. To address this issue, a solution is proposed in~\cite{cohen2023characterizing}, which aims to design a smooth controller with quantified robustness. This is achieved by recognizing that Sontag's formula is derived from solving an algebraic equation defined by an implicit function~\cite{sontag2013mathematical}. The authors of~\cite{cohen2023characterizing} suggest constructing a tunable implicit function to adjust the robustness of Sontag's formula while ensuring smoothness by satisfying specific conditions. However, the solution in~\cite{cohen2023characterizing} has two primary drawbacks. First, it requires formulating a set of ordinary differential equations and employing set invariance theory to derive tunable universal formulas. This process complicates the construction of the tunable implicit function, and the complexities of set invariance theory often hinder the analysis and understanding. Second, the tunable implicit function does not explicitly or intuitively show how it regulates safety, smoothness, and robustness.

In this paper, with the same objective to~\cite{cohen2023characterizing}, we aim to derive a universal formula to address a safety-critical task with desirable smoothness and robustness. Inspired by our prior work~\cite{li2024unifying}, which extends Lin-Sontag's formula~\cite{Norm_Bounded} to address a stabilizing control problem, we generalize the idea and propose a tunable universal formula approach to achieve the objective in this paper. As a result, we propose a novel and general framework for deriving a tunable universal formula, which improves the performance of QP-synthesized controllers in terms of robustness and smoothness. Moreover, it helps reduce the conservatism typically associated with Sontag's formula, offering a more flexible and versatile control solution. Compared to the approach in~\cite{cohen2023characterizing}, instead of starting from a tunable implicit function, we begin with a controller that has a predetermined structure (i.e., Sontag's formula in this paper) and incorporate a~\textit{tunable term} to regulate the performance of a controller. Compared to the results in~\cite{cohen2023characterizing}, our approach presents two main advantages over the results in~\cite{cohen2023characterizing}: (i) Simplicity. Unlike the first drawback identified in \cite{cohen2023characterizing}, we only need to ensure that the tunable term meets specified conditions and validity ranges; and (ii) Interpretability. The tunable term in our formula reflects its impact on safety and robustness, offering better interpretability compared to the approach in~\cite{cohen2023characterizing}. Additionally, we highlight that the tunable universal formula proposed in this paper can be utilized to address a safety-critical control problem with a norm-bounded input constraint. We demonstrate that, with a simple adaptation to the tunable term's validity range, the controller's performance can still be regulated in terms of safety, smoothness, and robustness, while satisfying the norm-bounded input constraint.

The remainder of this paper is organized as follows. Section~\ref{Preliminaries} presents the mathematical preliminaries and problem formulation. Section~\ref{Main_Results} derives the tunable universal formula and analyzes its safety, smoothness, and robustness properties. We also extend it to safety-critical control with norm-bounded inputs. Section~\ref{Simulation} demonstrates its effectiveness through a two-link manipulator safe tracking simulation. Section~\ref{Conclusions} concludes the paper. 
\section{Preliminaries and Problem Statement}\label{Preliminaries}
Consider a control-affine system
	\begin{equation}\label{Affine_Control_System}
	    \dot{\mathbf{x}}=\mathbf{f}(\mathbf{x})+\mathbf{g}(\mathbf{x})\mathbf{u},
	\end{equation}
where $\mathbf{x}(t)\in\mathbb{R}^{n}$ is the state, $\mathbf{u}(t)\in\mathbb{R}^{m}$ is the control input, and $\mathbf{f}:\mathbb{R}^{n}\rightarrow\mathbb{R}^{n}$ and $\mathbf{g}:\mathbb{R}^{n}\rightarrow\mathbb{R}^{n\times m}$ are smooth functions (i.e., $C^{\infty}$ differentiable functions). Given a Lipschitz continuous state-feedback controller $\mathbf{k}:\mathbb{R}^{n}\rightarrow\mathbb{R}^{m}$, the closed-loop system dynamics are:
\begin{equation}\label{state_feedback_dyanmics}
\dot{\mathbf{x}} = \mathbf{f}_{\text{cl}}(\mathbf{x})\triangleq\mathbf{f}(\mathbf{x})+\mathbf{g}(\mathbf{x}) \mathbf{k}(\mathbf{x}).
\end{equation}
Since the functions $\mathbf{f}$ and $\mathbf{g}$ are assumed to be smooth, and $\mathbf{k}$ is Lipschitz continuous, the function $\mathbf{f}_{\text{cl}}(\mathbf{x})$ is also Lipschitz continuous. Consequently, for any initial condition $\mathbf{x}_{0}=\mathbf{x}(\mathbf{0})\in\mathbb{R}^{n}$, there exists a time interval $I(\mathbf{x}_{0})=[0,t_{\max})$ such that $\mathbf{x}(t)$ is the unique solution to~\eqref{state_feedback_dyanmics} on $I(\mathbf{x}_{0})$. 
\subsection{CBFs and Universal Formula for Safety Control}
Consider a closed set $\mathcal{C}\subset\mathbb{R}^{n}$ as the $0$-superlevel set of a  smooth function $h:\mathbb{R}^{n}\rightarrow\mathbb{R}$, which is defined as
\begin{equation}\label{Invariant_Set}
		\begin{aligned}
			\mathcal{C} & \triangleq\left\{\mathbf{x}\in \mathbb{R}^{n}: h(\mathbf{x}) \geq 0\right\}, \\
			\partial \mathcal{C} & \triangleq\left\{\mathbf{x}\in\mathbb{R}^{n}: h(\mathbf{x})=0\right\}, \\
			\operatorname{Int}(\mathcal{C}) & \triangleq\left\{\mathbf{x}\in \mathbb{R}^{n}: h(\mathbf{x})>0\right\},
		\end{aligned}
\end{equation}
where we assume that $\mathcal{C}$ is nonempty and has no isolated points, that is, $\operatorname{Int}(\mathcal{C}) \neq \emptyset$ and $\bar{\operatorname{Int}(\mathcal{C})}=\mathcal{C}$.
\begin{Def}
    (Forward Invariance and Safety) A set $\mathcal{C}\subset\mathbb{R}^{n}$ is~\textit{forward invariant} if for every $\mathbf{x}_{0}\in\mathcal{C}$, the solution to~\eqref{state_feedback_dyanmics} satisfies $\mathbf{x}(t)\in\mathcal{C}$ for all $t\in I(\mathbf{x}_{0})$. The system is~\textit{safe} on the set $\mathcal{C}$ if the set $\mathcal{C}$ is forward invariant. 
\end{Def}

\begin{Def}
(Extended class $\mathrm{K}$ function, $\mathrm{K}_{e}$). A continuous function $\alpha:(-b,a)\rightarrow\mathbb{R}$, with $a,b>0$, is an extended class $\mathrm{K}$ function ($\alpha\in\mathrm{K}_{e}$) if $\alpha(0)=0$, and $\alpha$ is strictly monotonically increasing.
\end{Def}
\begin{Def}\label{CBF_Def}
		(CBFs~\cite{Zeroing_CBF}). Let $\mathcal{C}\subset\mathbb{R}^{n}$ be the $0$-superlevel set of a smooth function $h:\mathbb{R}^{n}\rightarrow\mathbb{R}$ which is defined by \eqref{Invariant_Set}. Then $h$ is a CBF for \eqref{Affine_Control_System} if there exists a smooth function $\beta\in\mathrm{K}_{e}$ such that, for all $\mathbf{x}\in\mathbb{R}^{n}$, there exists a control input $\mathbf{u}\in\mathbb{R}^{m}$ satisfying
		\begin{equation}\label{CBF_Condition}
    \begin{split}
        &c(\mathbf{x})+\mathbf{d}(\mathbf{x})\mathbf{u}\geq 0,
    \end{split}
\end{equation}
where $c(\mathbf{x})=L_{\mathbf{f}} h(\mathbf{x})+\beta({h(\mathbf{x})})$, $\mathbf{d}(\mathbf{x})=L_{\mathbf{g}} h(\mathbf{x})$, $L_{\mathbf{f}} h(\mathbf{x})\triangleq\frac{\partial h(\mathbf{x})}{\partial\mathbf{x}}\mathbf{f}(\mathbf{x})$ and $L_{\mathbf{g}} h(\mathbf{x})\triangleq\frac{\partial h(\mathbf{x})}{\partial\mathbf{x}}\mathbf{g}(\mathbf{x})$ denote the Lie derivatives along $\mathbf{f}$ and $\mathbf{g}$, respectively. 
\end{Def}

The CBF condition~\eqref{CBF_Condition} limits the choice of admissible control inputs. Considering its dependency on $\mathbf{x}$, we write the admissible input as a set-valued function $\mathcal{S}_{\mathrm{h}}:\mathbb{R}^{n}\rightrightarrows \mathbb{R}^m$, $\mathcal{S}_{\mathrm{h}}(\mathbf{x})\triangleq\{\mathbf{u}\in\mathbb{R}^{m}: \text{CBF constraint}~\eqref{CBF_Condition}\}$.
\begin{figure}[tp]
 \centering
    \makebox[0pt]{%
    \includegraphics[width=0.28\textwidth]{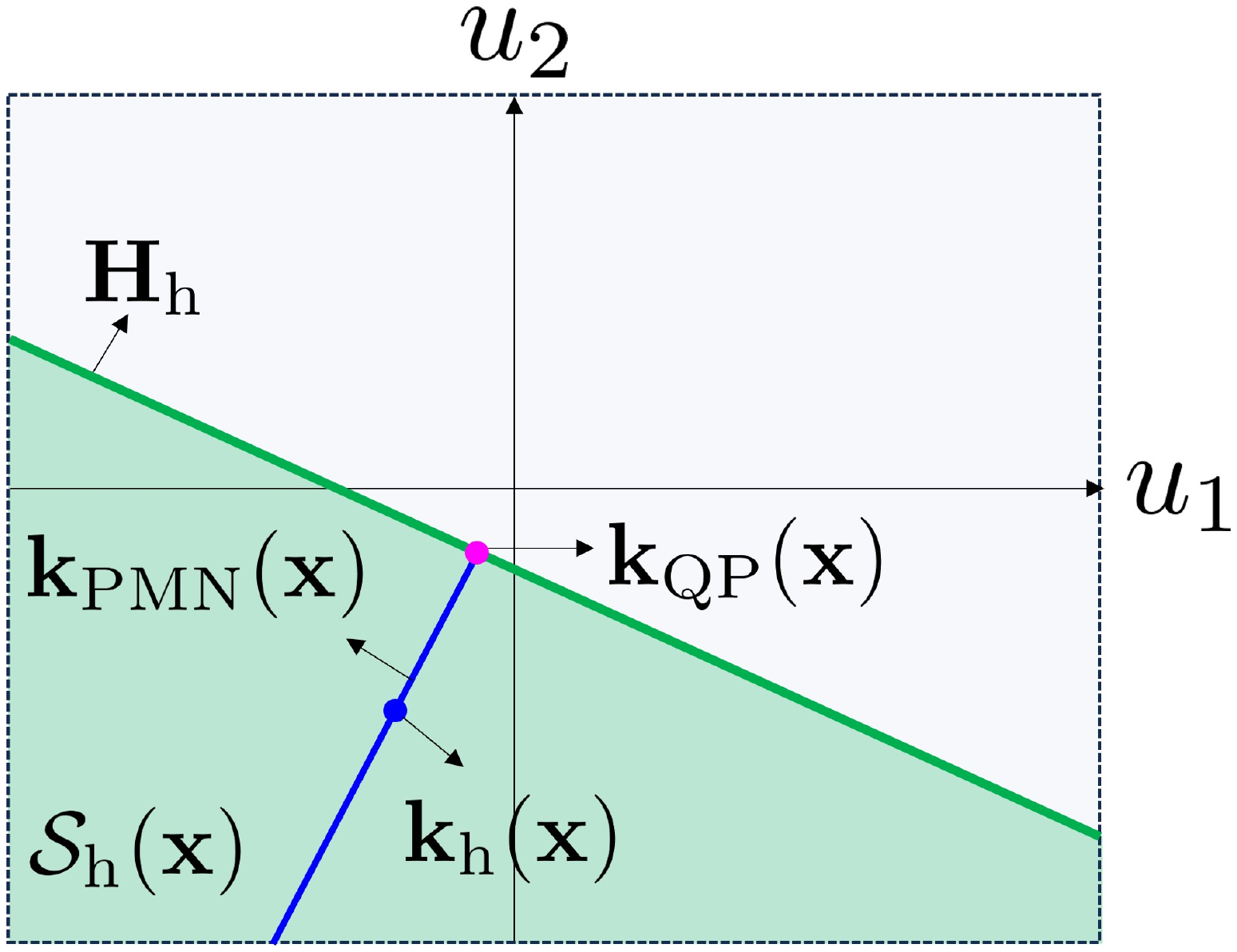}}
    \caption{A graphical illustration of universal formulas featuring different $\Gamma(\mathbf{x})$: the entire area corresponds to a 2-D control space $\mathbf{u}\in\mathbb{R}^{2}$; the admissible control input set $\mathcal{S}_{\mathrm{h}}(\mathbf{x})$ is represented by the green area; the green line represents the boundary of $\mathcal{S}_{\mathrm{h}}(\mathbf{x})$, forming a half-space denoted as $\mathbf{H}_{h}$; the blue ray indicates the control law $\mathbf{k}_{\mathrm{PMN}}(\mathbf{x})$ with various available choices for $\Gamma(\mathbf{x})$; the magenta point corresponds to the QP-synthesized controller $\mathbf{k}_{\mathrm{QP}}(\mathrm{x})$; and the blue point denotes Sontag's formula $\mathbf{k}_{\mathrm{Stg}}(\mathbf{x})$, which is constructed with $\Gamma_{\mathrm{Stg}}(\mathbf{x})$.}
    \label{CBF_Sontag_Universal}
\end{figure}
A graphical illustration of $\mathcal{S}_{\mathrm{h}}(\mathbf{x})$ in a 2-D control input space $\mathbf{u}\in\mathbb{R}^{2}$ is provided in Fig.~\ref{CBF_Sontag_Universal}. For each $\mathbf{x}\in\mathbb{R}^{n}$, $S_{h}(\mathbf{x})$ is a half-space. Moreover, the hyperplane corresponding to the condition~\eqref{CBF_Condition} is defined as $\mathbf{H}_{h}(\mathbf{x})=\{\mathbf{u}\in\mathbb{R}^{m} \mid c(\mathbf{x})+\mathbf{d}(\mathbf{x}) \mathbf{u}=0\}$, which is illustrated by the blue ray in Fig.~\ref{CBF_Sontag_Universal}.
\begin{Rmk}\label{Rmk_Strict}
    To ensure the Lipschitz continuity of the tunable universal formula, it is crucial to modify the standard CBF condition~\eqref{CBF_Condition} into a strict inequality, see~\cite{Safety_Margin} and~\cite[Remarks 1 and 3]{cohen2024safety}.
\end{Rmk}
\begin{Thm}\label{State_Feedback_Safety}
(\hspace{-0.005cm}\cite{Zeroing_CBF}) Given $\mathcal{C}\subset\mathbb{R}^{n}$ defined as the $0$-superlevel set of a smooth function $h:\mathbb{R}^{n}\rightarrow\mathbb{R}$, and $h$ serves as a CBF for~\eqref{Affine_Control_System} on $\mathcal{C}$, then any Lipschitz continuous feedback controller $\mathbf{k}:\mathbb{R}^{n}\rightarrow\mathbb{R}^{m}$ satisfying the condition~\eqref{CBF_Condition} ensures the safety of the closed-loop system~\eqref{Affine_Control_System} with respect to the set $\mathcal{C}$.
\end{Thm}
    
To provide freedom to regulate the control performance, the tightened CBF condition, which is a sufficient condition of~\eqref{CBF_Condition}, is introduced as follows: 
\begin{equation}\label{Tightened_CBF}
    c(\mathbf{x})+\mathbf{d}(\mathbf{x})\mathbf{u}\geq\Gamma(\mathbf{x}),
\end{equation}
where $\Gamma:\mathbb{R}^{n}\rightarrow\mathbb{R}_{\geq 0}$ is a positive semidefinite function satisfying $c(\mathbf{x})\geq\Gamma(\mathbf{x})$ whenever $\mathbf{d}(\mathbf{x})=\mathbf{0}$.
\begin{Rmk}
Note that, in~\eqref{Tightened_CBF}, every choice of $\Gamma(\mathbf{x})$ defines a specific CBF. Specifically, when $\Gamma(\mathbf{x}) = 0$, Eq.~\eqref{Tightened_CBF} represents the standard CBF condition, which requires  $c(\mathbf{x})\geq 0$ whenever $\mathbf{d}(\mathbf{x})=\mathbf{0}$. When we select positive values for $\Gamma(\mathbf{x})$, it leads to a more robust CBF~\cite{cohen2022robust}, which requires that $c(\mathbf{x}) \geq \Gamma(\mathbf{x})$ whenever $\mathbf{d}(\mathbf{x}) = \mathbf{0}$. In this paper, we assume the existence of such a robust CBF, which is a commonly used assumption~\cite{ISSf,Safety_Margin}. However, the construction of a robust CBF is not addressed here, as we focus primarily on the controller design. For those interested in constructing robust CBFs, we recommend referring to~\cite{cohen2022robust} and the references therein.
\end{Rmk}

Based on Theorem~\ref{State_Feedback_Safety}, a feedback control law $\mathbf{k}_{\mathrm{PMN}}(\mathbf{x})$ can be synthesized by solving the following pointwise min-norm (PMN) optimization problem:
\begin{equation}\label{PMN_Safety}
\begin{aligned}
\mathbf{k}_{\mathrm{PMN}}(\mathbf{x})=\underset{\mathbf{u} \in \mathbb{R}^m}{\arg \min } & \,\,\frac{1}{2}\left\|\mathbf{u}\right\|^2 \\
\text {s.t.} &  \,\,\text{Condition}~\eqref{Tightened_CBF},
\end{aligned}
\end{equation}
where $\|\mathbf{u}\|$ denotes the $2$-norm of the control input.
\begin{Lem}\label{PMN_Solution}
The solution to~\eqref{PMN_Safety} can be explicitly expressed as
\begin{equation}\label{QP_CBF_Solution}
    \begin{split}
      \mathbf{k}_{\mathrm{PMN}}(\mathbf{x})=\lambda_{\mathrm{PMN}}(c(\mathbf{x})-\Gamma(\mathbf{x}),\|\mathbf{d}( \mathbf{x}) \|^{2}) \mathbf{d}(\mathbf{x})^{\top},
    \end{split}
\end{equation}
where 
\begin{equation}\label{Relu_Function}
\lambda_{\mathrm{PMN}}(c, d):= \begin{cases}0, & d=0, \\ \operatorname{ReLU}(-c / d), & d>0,\end{cases}
\end{equation}
and $\operatorname{ReLU}(y):=\max \{0, y\}$.
\end{Lem}
\begin{proof}
The solution to~\eqref{PMN_Safety} can be obtained using the projection theorem. More details can be found in~\cite{Universal_Formula_Gain_Margin_1}.    
\end{proof}

In Fig.~\ref{CBF_Sontag_Universal}, the control law $\mathbf{k}_{\mathrm{PMN}}(\mathbf{x})$ corresponds to the blue ray. Each specific choice of $\Gamma(\mathbf{x})$ corresponds to a distinct point on the blue ray. Actually, the selection of the function $\Gamma(\mathbf{x})$ significantly impacts the performance of $\mathbf{k}_{\mathrm{PMN}}(\mathbf{x})$ for safety control, such as smoothness, robustness, inverse optimality, etc. For instance, by setting $\Gamma(\mathbf{x})=\Gamma_{\mathrm{QP}}(\mathbf{x})=0$ in the constraint of equation~\eqref{PMN_Safety}, as per Equation~\eqref{QP_CBF_Solution}, it yields a control law:
\begin{equation}\label{QP_Synthesized}
    \begin{split}
      \mathbf{k}_{\mathrm{QP}}(\mathbf{x})=\lambda_{\mathrm{PMN}}(c(\mathbf{x}),\|\mathbf{d}( \mathbf{x}) \|^{2}) \mathbf{d}(\mathbf{x})^{\top},
    \end{split}
\end{equation}
where $\mathbf{k}_{\mathrm{QP}}(\mathbf{x})$ is commonly known as a QP-synthesized controller~\cite{cohen2023characterizing}. As shown in Fig.~\ref{CBF_Sontag_Universal}, $\mathbf{k}_{\mathrm{QP}}(\mathbf{x})$ corresponds to the magenta point. 
Alternatively, one can choose
\begin{equation}\label{Gamma_Choice}
    \Gamma(\mathbf{x})=\Gamma_{\mathrm{Stg}}(\mathbf{x})=\sqrt{c^{2}(\mathbf{x})+s(\|\mathbf{d}(\mathbf{x}) \|^{2})\|\mathbf{d}(\mathbf{x}) \|^{2}},
\end{equation}
where $s:\mathbb{R}\rightarrow\mathbb{R}$ is smooth and satisfies $s(0)=0$ and $s(d)>0$ for all $d\neq 0$. This gives
\begin{equation}\label{Scaling_Condition}
    c(\mathbf{x})+\mathbf{d}(\mathbf{x}) \mathbf{u}\geq \Gamma_{\mathrm{Stg}}(\mathbf{x}).
\end{equation}
By applying condition~\eqref{Scaling_Condition} to the constraint in~\eqref{PMN_Safety} and leveraging Lemma~\ref{PMN_Solution}, we obtain Sontag's formula for safety-critical control, which is given by:
\begin{equation}\label{CBF_Universal_Law}
\begin{split}
\mathbf{k}_{\mathrm{Stg}}(\mathbf{x})=\lambda_{\mathrm{Stg}}(c(\mathbf{x}),\|\mathbf{d}( \mathbf{x}) \|^{2}) \mathbf{d}(\mathbf{x})^{\top},
\end{split}
\end{equation}
where 
\begin{equation}\label{Stg_Function}
\lambda_{\mathrm{Stg}}(c, d):= \begin{cases}
                     0, & d=0, \\
                   \left(-c+\sqrt{c^{2}+s(d)d}\right)/d, & d>0.
                   \end{cases}
\end{equation}
As illustrated in Fig.~\ref{CBF_Sontag_Universal}, Sontag's formula for safety-critical control, i.e., $\mathbf{k}_{\mathrm{Stg}}(\mathbf{x})$ in~\eqref{CBF_Universal_Law}, corresponds to the blue point on the blue ray in Fig.~\ref{CBF_Sontag_Universal}. This indicates that Sontag's formula is a special solution to the PMN optimization in~\eqref{PMN_Safety} since the control law $\mathbf{k}_{\mathrm{Stg}}(\mathbf{x})$ is obtained from~\eqref{QP_CBF_Solution} by setting the specific function as $\Gamma(\mathbf{x}) = \Gamma_{\mathrm{Stg}}(\mathbf{x})$.
\subsection{Safety, Smoothness, and Robustness}
In this subsection, we examine the safety, smoothness, and robustness of the QP-synthesized controller and Sontag's formula.
\subsubsection{Safety}\label{Safety_Exist}
Substituting $\mathbf{k}_{\mathrm{QP}}(\mathrm{x})$ from~\eqref{QP_Synthesized} and $\mathbf{k}_{\mathrm{Stg}}(\mathrm{x})$ from~\eqref{CBF_Universal_Law} into~\eqref{CBF_Condition} shows that both controllers satisfy the CBF condition~\eqref{CBF_Condition}. Therefore, according to Theorem~\ref{State_Feedback_Safety}, we conclude that both the QP-synthesized controller and Sontag's formula ensure the safety of the system~\eqref{Affine_Control_System}. 
\subsubsection{Smoothness}\label{Smoothness_exist}
The function $\lambda_{\mathrm{PMN}}(c, d)$, as defined in~\eqref{Relu_Function}, may be non-smooth due to its piecewise nature and the inherent non-smoothness of the $\operatorname{ReLU}(\cdot)$ function. Therefore, the QP-synthesized controller $\mathbf{k}_{\mathrm{QP}}(\mathbf{x})$ given in~\eqref{QP_Synthesized} may be non-smooth, as discussed in~\cite{cohen2023characterizing}. In contrast, Sontag's formula for safety-critical control, i.e., $\mathbf{k}_{\mathrm{Stg}}(\mathrm{x})$ in~\eqref{CBF_Universal_Law}, is shown to be smooth~\cite{cohen2023characterizing}, with its smoothness being proved by using the approach in~\cite[Proposition 5.9.10]{sontag2013mathematical}.
\subsubsection{Robustness}\label{Robustness_exist}
The QP-synthesized controller $\mathbf{k}_{\mathrm{QP}}(\mathbf{x})$ could exhibit non-robust behavior when it is subjected to an input disturbance $\mathbf{w}(\mathbf{x})$. Specifically, the CBF condition is ensured by the QP-synthesized controller $\mathbf{k}_{\mathrm{QP}}(\mathbf{x})$, i.e., $c(\mathbf{x})+\mathbf{d}(\mathbf{x})\mathbf{k}_{\mathrm{QP}}(\mathbf{x})\geq 0$, while an input disturbance $\mathbf{w}(\mathbf{x})$ will give a control input $\mathbf{k}_{\mathrm{QP-D}}(\mathbf{x})=\mathbf{k}_{\mathrm{QP}}(\mathbf{x})+\mathbf{w}(\mathbf{x})$. By substituting $\mathbf{k}_{\mathrm{QP-D}}(\mathbf{x})$ into~\eqref{CBF_Condition}, we obtain the following condition:
\begin{equation}\label{QP_Disturbed}
    c(\mathbf{x})+\mathbf{d}(\mathbf{x})\mathbf{k}_{\mathrm{QP-D}}(\mathbf{x})\geq \mathbf{d}(\mathbf{x})\mathbf{w}(\mathbf{x}).
\end{equation}
Since the term $\mathbf{d}(\mathbf{x})\mathbf{w}(\mathbf{x})$ in~\eqref{QP_Disturbed} can be negative, it is possible to have $c(\mathbf{x}) + \mathbf{d}(\mathbf{x})\mathbf{k}_{\mathrm{QP-D}}(\mathbf{x}) < 0$. To solve this problem, the system~\eqref{Affine_Control_System} may become unsafe due to the input disturbance $\mathbf{w}(\mathbf{x})$. As shown in Fig.~\ref{CBF_Sontag_Universal}, the QP-synthesized controller $\mathbf{k}_{\mathrm{QP}}(\mathbf{x})$ lies on the boundary of the set $\mathcal{S}_{\mathrm{h}}(\mathbf{x})$. In this scenario, a small input disturbance $\mathbf{w}(\mathbf{x})$ may cause $\mathbf{k}_{\mathrm{QP-D}}(\mathbf{x}) = \mathbf{k}_{\mathrm{QP}}(\mathbf{x}) + \mathbf{w}(\mathbf{x})$ to fall outside of $\mathcal{S}_{\mathrm{h}}(\mathbf{x})$, leading to the system~\eqref{Affine_Control_System} being unsafe.

In contrast, Sontag's formula $\mathbf{k}_{\mathrm{Stg}}(\mathbf{x})$ is more robust than the QP-synthesized controller $\mathbf{k}_{\mathrm{QP}}(\mathbf{x})$ in safety-critical control. In particular, by substituting Sontag's formula with input disturbance, i.e., $\mathbf{k}_{\mathrm{Stg-D}}(\mathbf{x})=\mathbf{k}_{\mathrm{Stg}}(\mathbf{x})+\mathbf{w}(\mathbf{x})$, into~\eqref{CBF_Condition} and applying the conditions in~\eqref{Gamma_Choice} and~\eqref{Scaling_Condition}, we can derive:
\begin{equation}\label{Sontag_Disturbed}
\begin{split}
    c(\mathbf{x})+\mathbf{d}(\mathbf{x})\mathbf{k}_{\mathrm{Stg-D}}(\mathbf{x})\geq \Gamma_{\mathrm{Stg}}(\mathbf{x})+\mathbf{d}(\mathbf{x})\mathbf{w}(\mathbf{x}).
\end{split}   
\end{equation}
By comparing~\eqref{QP_Disturbed} and~\eqref{Sontag_Disturbed}, we note that the control input $\mathbf{k}_{\mathrm{Stg-D}}$ in~\eqref{Sontag_Disturbed} is more robust to ensure the safety of the system~\eqref{Affine_Control_System} due to the existence of the positive definite function $\Gamma_{\mathrm{Stg}}(\mathbf{x})$ in~\eqref{Sontag_Disturbed}. Geometrically, as depicted in Fig.~\ref{CBF_Sontag_Universal}, Sontag's formula $\mathbf{k}_{\mathrm{Stg}}(\mathbf{x})$ lies within the set $\mathcal{S}_h(\mathbf{x})$, while the control law $\mathbf{k}_{\mathrm{QP}}(\mathbf{x})$ is on the boundary of $\mathcal{S}_h(\mathbf{x})$. Therefore, Sontag's formula is less likely to fall outside the set $\mathcal{S}_h(\mathbf{x})$ in the presence of the input disturbance $\mathbf{w}(\mathbf{x})$.
\begin{Rmk}\label{Rmk_Conservatism}
While Sontag's formula is robust against disturbance (compared to the QP-synthesized controller), applying Sontag's formula $\mathbf{k}_{\mathrm{Stg}}(\mathbf{x})$ in~\eqref{CBF_Universal_Law} can lead to conservative behavior for safety-critical control. With the choice in~\eqref{Gamma_Choice}, Sontag's formula $\mathbf{k}_{\mathrm{Stg}}(\mathbf{x})$ satisfies a stricter safety condition, i.e., the tightened CBF condition~\eqref{Scaling_Condition} with $\Gamma(\mathbf{x})=\Gamma_{\mathrm{Stg}}(\mathbf{x})$, compared to the standard CBF condition~\eqref{CBF_Condition}. Therefore, Sontag's formula exhibits a more conservative safety control performance compared to QP-synthesized controller. As demonstrated by the reach-avoid example from~\cite[Example 1]{cohen2023characterizing}, Sontag's formula can result in fairly conservative behavior, where the safety-critical control dominates the behavior in goal reaching.
\end{Rmk}
\subsection{Problem Statement}\label{Problem_Statement}
The above discussions suggest that both the QP-synthesized controller and Sontag's formula may not achieve the desired performance requirements when applied to safety-critical control scenarios. 
Given that the selection of $\Gamma(\mathbf{x})$ significantly influences the properties of the resulting controller, we aim to find a suitable $\Gamma(\mathbf{x})$ to derive alternative universal formulas that satisfy safety, smoothness, and robustness requirements. Therefore, the research problem of this paper is formally stated as follows.

\textit{\textbf{Problem Statement:}} Design a suitable $\Gamma(\mathbf{x})$ to construct a universal formula addressing a safety-critical control task, which should feature desirable properties, including safety, smoothness, and robustness.
\section{Main Results}\label{Main_Results}
In this section, a tunable universal formula is derived, and the properties of the tunable universal formula related to safety guarantees, smoothness, and robustness, are studied. Furthermore, we extend the tunable universal formula to the application of input-constrained safety-critical control.
\subsection{Tunable Universal Formulas for Safety-Critical Control}\label{Tun_Without_Input}
To address the problem stated in Section~\ref{Problem_Statement}, our solution is to design a function of the form $\Gamma(\mathbf{x})=\kappa(\mathbf{x})\Gamma_{\mathrm{Basis}}(\mathbf{x})$. Here, $\kappa:\mathbb{R}^{n}\rightarrow\mathbb{R}$ serves as a state-dependent tunable term, while $\Gamma_{\mathrm{Basis}}(\mathbf{x})$ represents a basis function that influences the properties of the derived universal formulas. Considering that Sontag's formula maintains numerous desirable properties, e.g., smoothness~\cite{cohen2023characterizing}, inverse optimality~\cite{Inverse_Optimality}, and strict safety guarantees~\cite{li2023graphical}, we set $\Gamma_{\mathrm{Basis}}(\mathbf{x})=\Gamma_{\mathrm{Stg}}(\mathbf{x})$ to derive a Sontag-like formula. Then the constraint in~\eqref{PMN_Safety} is reformulated as: 
\begin{equation}\label{Scaling_Condition_kappa}
    c(\mathbf{x})+\mathbf{d}(\mathbf{x}) \mathbf{u}\geq \kappa(\mathbf{x})\Gamma_{\mathrm{Stg}}(\mathbf{x}).
\end{equation}
\begin{Rmk}
In~\eqref{Scaling_Condition_kappa}, we set $\Gamma_{\mathrm{Basis}}(\mathbf{x}) = \Gamma_{\mathrm{Stg}}(\mathbf{x})$ and introduce a state-dependent term $\kappa(\mathbf{x})$ to adjust the robustness of the controller under the condition~\eqref{Scaling_Condition_kappa}. Notably, selecting alternative basis functions provides the designed controllers with different and crucial properties. For instance, in~\cite[Eq. (28)]{ISSf}, the authors choose $\Gamma_{\mathrm{Basis}}(\mathbf{x}) = \|\mathbf{d}(\mathbf{x})\|^{2}$ to account for input disturbances and establish a condition that ensures input-to-state safety of the system~\eqref{Affine_Control_System}. Moreover, a tunable term is introduced in the basis function to adjust the robustness in~\cite{alan2021safe}, which allows the design of a controller that balances robustness and safety. In contrast to~\cite{alan2021safe} and~\cite{ISSf}, our focus extends beyond robustness and safety to also emphasize smoothness. Thus, we set $\Gamma_{\mathrm{Basis}}(\mathbf{x}) = \Gamma_{\mathrm{Stg}}(\mathbf{x})$ to preserve the smoothness property of Sontag's formula.
\end{Rmk}

To ensure~\eqref{Scaling_Condition_kappa} hold for safety-critical control, two conditions must be satisfied: i) $\kappa(\mathbf{x})\Gamma_{\mathrm{Stg}}(\mathbf{x})\geq 0$, and ii) $c(\mathbf{x})\geq\kappa(\mathbf{x})\Gamma_{\mathrm{Stg}}(\mathbf{x})$ when $\mathbf{d}(\mathbf{x})=\mathbf{0}$. With conditions i) and ii), we determine the range of $\kappa(\mathbf{x})$ to be $0\leq\kappa(\mathbf{x})\leq 1$. Next, we substitute the constraint in~\eqref{PMN_Safety} with the condition~\eqref{Scaling_Condition_kappa}. With Lemma~\ref{PMN_Solution}, the following control law is obtained:
    \begin{equation}\label{PMN_Tunable_Universal_Formula}
    \begin{split}
      \bar{\mathbf{k}}_{\mathrm{Tun}}(\mathbf{x})=\bar{\lambda}_{\mathrm{Tun}}(c(\mathbf{x}),\|\mathbf{d}( \mathbf{x}) \|^{2},\kappa(\mathbf{x})) \mathbf{d}(\mathbf{x})^{\top},
    \end{split}
\end{equation}
where
\begin{equation}\label{PMN_Tunable_Function}
\begin{split}
      \bar{\lambda}_{\mathrm{Tun}}(c,d,\iota):=\left\{\begin{array}{ll}
      0,&\text{if}\,\, d= 0,\\
      \operatorname{ReLU}(\frac{-c+\iota\sqrt{c^2+s(d)d}}{d}),& \text{if}\,\, d\neq 0.    
\end{array}\right.  
\end{split}
\end{equation}
\begin{Rmk}
    For the control law $\bar{\mathbf{k}}_{\mathrm{Tun}}(\mathbf{x})$ defined in~\eqref{PMN_Tunable_Universal_Formula}, different choices of the tunable term $\kappa(\mathbf{x})$ will result in distinct control laws, each characterized by desirable smoothness, safety, and robustness properties. For instance, 
    when one sets $\kappa(\mathbf{x})=0$ for $\bar{\mathbf{k}}_{\mathrm{Tun}}(\mathbf{x})$, the QP-synthesized controller $\mathbf{k}_{\mathrm{QP}}(\mathbf{x})$ presented in~\eqref{QP_Synthesized} is obtained, and $\mathbf{k}_{\mathrm{QP}}(\mathbf{x})$ maybe non-smooth as discussed in Section~\ref{Smoothness_exist}. When $\kappa(\mathbf{x})=1$, the control law $\bar{\mathbf{k}}_{\mathrm{Tun}}(\mathbf{x})$ leads to Sontag's formula $\mathbf{k}_{\mathrm{Stg}}(\mathbf{x})$ as specified in~\eqref{CBF_Universal_Law}, which is a smooth function (cf. Section~\ref{Smoothness_exist}). Moreover, it is important to highlight that the selection of the tunable term clearly reflects its influence on both safety and robustness properties, offering better interpretability than~\cite{cohen2023characterizing}. Specifically, as the value of $\kappa(\mathbf{x})$ increases, the controller $\bar{\mathbf{k}}_{\mathrm{Tun}}(\mathbf{x})$ will become more robust to input disturbance and have more conservative safety-critical control behavior according to the condition~\eqref{Scaling_Condition_kappa} and similar analysis as in Section~\ref{Robustness_exist}. Motivated by the above analysis, we call $\bar{\mathbf{k}}_{\mathrm{Tun}}(\mathbf{x})$ a tunable universal formula.
\end{Rmk}
\subsubsection{Safety}\label{Safety_Guarantee_section}
We provide the following theorem to guarantee the safety of the system~\eqref{Affine_Control_System} by using the tunable universal formula given in~\eqref{PMN_Tunable_Universal_Formula}.
    \begin{Thm}\label{Safety_Guarantee_Thm}        
    Assume that $h:\mathbb{R}^{n}\rightarrow\mathbb{R}$ is a CBF.
Then the tunable universal formula $\bar{\mathbf{k}}_{\mathrm{Tun}}(\mathbf{x})$ provided
by~\eqref{PMN_Tunable_Universal_Formula}, with a Lipschitz continuous function $\kappa(\mathbf{x})\in\mathcal{K}_{\mathrm{SA}}$, 
\begin{equation}\label{Kappa_SA}
    \mathcal{K}_{\mathrm{SA}}=\{\kappa:\mathbb{R}^{n}\rightarrow\mathbb{R}| 0<\kappa(\mathbf{x})\leq 1\},
\end{equation}
ensures the safety of the closed-loop system~\eqref{Affine_Control_System}.
    \end{Thm}
\begin{proof}
    By substituting~\eqref{PMN_Tunable_Universal_Formula} into~\eqref{CBF_Condition}, it shows that the CBF condition~\eqref{CBF_Condition} is always satisfied for all $0\leq\kappa(\mathbf{x})\leq 1$. Next, as noted in Remark~\ref{Rmk_Strict}, we should exclude $\kappa(\mathbf{x})=0$ to ensure the Lipschitz continuity of CBF-based controllers, leading to the constraint $0 < \kappa(\mathbf{x}) \leq 1$. Finally, with Theorem~\ref{State_Feedback_Safety}, we conclude that the closed-loop system~\eqref{Affine_Control_System} (with $\bar{\mathbf{k}}_{\mathrm{Tun}}(\mathbf{x})$) is safe.
\end{proof}
\subsubsection{Smoothness}\label{Tunable_Uni_Smooth}
To ensure a proper choice of $\kappa(\mathbf{x})$ that leads to the tunable universal formula $\bar{\mathbf{k}}_{\mathrm{Tun}}(\mathbf{x})$ being a smooth control law, we first need to eliminate the influence of the function $\operatorname{ReLU}(\cdot)$ as appeared in~\eqref{PMN_Tunable_Function}. To achieve this, we choose a function $\iota$ that satisfies $\iota> c/\sqrt{c^2+s(d)d}$, which gives $-c+\iota\sqrt{c^2+s(d)d}> 0$. This choice simplifies the function $\bar{\lambda}_{\mathrm{Tun}}(c,d,\iota)$ (defined in~\eqref{PMN_Tunable_Function}) to 
\begin{equation}\label{Tunable_Function}
\begin{split}
      \lambda_{\mathrm{Tun}}(c,d,\iota):=\left\{\begin{array}{ll}
      0,&\text{if}\,\, d= 0,\\
       \frac{-c+\iota\sqrt{c^2+s(d)d}}{d},& \text{if}\,\, d\neq 0.       
\end{array}\right.  
\end{split}
\end{equation}
For the tunable universal formula $\bar{\mathbf{k}}_{\mathrm{Tun}}(\mathbf{x})$, the above discussion essentially requires that $\kappa(\mathbf{x}) > \frac{c(\mathbf{x})}{\Gamma_{\mathrm{Stg}}(\mathbf{x})}$. Moreover, since the smoothness discussions relate to a safety-critical controller, we need to incorporate the condition on $\kappa(\mathbf{x})$ in Theorem~\ref{Safety_Guarantee_Thm}, specifically $0 < \kappa(\mathbf{x}) \leq 1$, which results in
\begin{equation}\label{Tunable_Universal_Formula}
    \begin{split}
      \mathbf{k}_{\mathrm{Tun}}(\mathbf{x})=\lambda_{\mathrm{Tun}}(c(\mathbf{x}),\|\mathbf{d}( \mathbf{x}) \|^{2},\kappa(\mathbf{x})) \mathbf{d}(\mathbf{x})^{\top},    
    \end{split}
\end{equation}
where $\kappa(\mathbf{x})$ is a tunable term drawn from the set $\mathcal{K}_{\mathrm{SM}}$, which is defined as follows:
\begin{equation}\label{Kappa_SM}
    \mathcal{K}_{\mathrm{SM}}=\{\kappa:\mathbb{R}^{n}\rightarrow\mathbb{R}|\max\left(\frac{c(\mathbf{x})}{\Gamma_{\mathrm{Stg}}(\mathbf{x})},0\right)<\kappa(\mathbf{x})\leq 1\}
\end{equation}
\begin{Rmk}\label{safe_filter_tun}
     Note that one can also adapt the tunable universal formula~\eqref{Tunable_Universal_Formula} to a safety filter framework~\cite{ames2019control}, which aims to make a minimal modification to a predefined nominal control $\mathbf{k}_{\mathrm{d}}(\mathbf{x})$. In this case, the cost function in~\eqref{PMN_Safety} is adjusted to $\frac{1}{2}\|\mathbf{u}-\mathbf{k}_{\mathrm{d}}\|^{2}$. By following the same derivation routine for obtaining~\eqref{Tunable_Universal_Formula}, the tunable universal formula for a safety filter is acquired, which is expressed as: 
    \begin{equation}\label{SF_Tunable Universal_Formula}
    \begin{split}
\mathbf{k}_{\mathrm{SF}}(\mathbf{x})=\lambda_{\mathrm{Tun}}(\bar{c}(\mathbf{x}),\|\mathbf{d}( \mathbf{x}) \|^{2},\kappa(\mathbf{x})) \mathbf{d}(\mathbf{x})^{\top}+\mathbf{k}_{\mathrm{d}}, 
    \end{split}
\end{equation}
where $\kappa(\mathbf{x})\in\mathcal{K}_{\mathrm{SF}}$, $\mathcal{K}_{\mathrm{SF}}=\{\kappa:\mathbb{R}^{n}\rightarrow\mathbb{R}|
\max(\frac{\bar{c}(\mathbf{x})}{\bar{\Gamma}_{\mathrm{Stg}}(\mathbf{x})},0)<\kappa(\mathbf{x})\leq 1\}$, $\bar{c}(\mathbf{x})=c(\mathbf{x})+\mathbf{d}(\mathbf{x})\mathbf{k}_{\mathrm{d}}$, and $\bar{\Gamma}_{\mathrm{Stg}}(\mathbf{x})$ can be obtained by replacing $c(\mathbf{x})$ in~\eqref{Gamma_Choice} with $\bar{c}(\mathbf{x})$.
\end{Rmk}

The above discussion does not provide us a constructive approach for designing a controller with smoothness guarantees. To address this, we first define an open subset $\Phi=\{(c,d)\in\mathbb{R}^{2}|c> 0 \,\,\text{or} \,\, d>0 \}$ and introduce the following lemma.
\begin{Lem}\label{Implicit_Theorem}
Assume that $\zeta(c,d)$ is smooth on $\Phi$ and let $\iota(c,d)=(1-\zeta(c,d))\cdot\frac{c}{\sqrt{c^{2}+s(d)d}}+\zeta(c,d)$. The function $\lambda_{\mathrm{Tun}}(c,d,\iota(c,d))$ defined in~\eqref{Tunable_Function}  
is smooth on $\Phi$.
\end{Lem}
\begin{proof}
Firstly, we substitute $\iota(c,d)=(1-\zeta(c,d))\cdot\frac{c}{\sqrt{c^{2}+s(d)d}}+\zeta(c,d)$ into~\eqref{Tunable_Function}, which gives 
\begin{equation}\label{Tunable_Function_Modify}
\begin{split}
      \lambda_{\mathrm{Tun}}(c,d,\iota(c,d))=\zeta(c,d)\cdot\lambda_{\mathrm{Stg}}(c, d).
\end{split}
\end{equation}
Note that one can prove that $\lambda_{\mathrm{Stg}}(c, d)$ is smooth on $\Phi$ by following the proof in~\cite[Proposition 5.9.10]{sontag2013mathematical}. Furthermore, due to  $\zeta(c,d)$ being a smooth function on $\Phi$ by assumption, we can infer from~\eqref{Tunable_Function_Modify} that $\lambda_{\mathrm{Tun}}(c,d,\iota(c,d))$ is also smooth on $\Phi$.
\end{proof}

Generally, there are various choices of $\iota(c,d)$ to ensure the smoothness of the function $\lambda_{\mathrm{Tun}}(c,d,\iota(c,d))$ given in~\eqref{Tunable_Universal_Formula}. For instance, setting $\iota(c,d)=1$ and $\iota(c,d)=0.5\cdot\frac{c}{\sqrt{c^{2}+s(d)d}}+0.5$, one can easily verify that the function $\lambda_{\mathrm{Tun}}(c,d,\iota(c,d))$ is smooth by following the proof in~\cite[Proposition 5.9.10]{sontag2013mathematical}. As for Lemma~\ref{Implicit_Theorem}, it provides a family of feasible $\iota(c,d)$ that ensures the smoothness of the function $\lambda_{\mathrm{Tun}}(c,d,\iota(c,d))$. 
\begin{Thm}\label{Property_for_U_with_Kappa}
Assume that $h:\mathbb{R}^{n}\rightarrow\mathbb{R}$ is a CBF and $\eta:\mathbb{R}^{2}\rightarrow\mathbb{R}$ is a smooth function on $\Phi$.
For the tunable universal formula provided in~\eqref{Tunable_Universal_Formula}, if one constructs $\kappa(\mathbf{x})$ using $\kappa(\mathbf{x})=(1-\eta)\cdot\frac{c(\mathbf{x})}{\Gamma_{\mathrm{Stg}}(\mathbf{x})}+\eta$ and ensures $\kappa(\mathbf{x})\in\mathcal{K}_{\mathrm{SM}}$, then the resulting control law is smooth and ensures the closed-loop system~\eqref{Affine_Control_System} is safe $\forall\mathbf{x}\in\mathbb{R}^{n}$.
\end{Thm}
\begin{proof}
Firstly, due to the condition $\kappa(\mathbf{x})\in\mathcal{K}_{\mathrm{SM}}$ and Theorem~\ref{Safety_Guarantee_Thm}, we can deduce that the closed-loop system~\eqref{Affine_Control_System} is safe with a control law defined by a tunable universal formula. Next, according to~\eqref{CBF_Condition} with a strict inequality condition, we know that $c(\mathbf{x})> 0$ when $\mathbf{d}(\mathbf{x})=\mathbf{0}$, and  $\|\mathbf{d}(\mathbf{x})\|^{2}>\mathbf{0}$ for $\mathbf{d}(\mathbf{x})\neq\mathbf{0}$. Consequently, we can verify that $(c(\mathbf{x}),\|\mathbf{d}(\mathbf{x})\|^2)\in\Phi$. Note that  $\eta(c(\mathbf{x}),\|\mathbf{d}(\mathbf{x})\|^{2})$ is a smooth function on $\Phi$ by assumption.\footnote{To simplify notation, we denote $\eta(c(\mathbf{x}),\|\mathbf{d}(\mathbf{x})\|^{2})$ by $\eta$ thereafter.} Then, due to $\kappa(\mathbf{x})=(1-\eta)\cdot\frac{c(\mathbf{x})}{\Gamma_{\mathrm{Stg}}(\mathbf{x})}+\eta$, one can prove that the function $\lambda_{\mathrm{Tun}}(c(\mathbf{x}),\|\mathbf{d}( \mathbf{x}) \|^{2},\kappa(\mathbf{x}))$ is smooth on $\Phi$ according to Lemma~\ref{Implicit_Theorem}. Therefore, the tunable universal formula given in~\eqref{Tunable_Universal_Formula} is smooth on $\Phi$, and hence smooth $\forall\mathbf{x}\in\mathbb{R}^{n}$.
\end{proof}

\begin{Rmk}\label{eta_range}
In Theorem~\ref{Property_for_U_with_Kappa}, one should note that the condition $\kappa(\mathbf{x})\in\mathcal{K}_{\mathrm{SM}}$ necessitates the smooth function $\eta\in\Xi$, where $\Xi:=\{\eta:\Phi\rightarrow\mathbb{R}|\max\left(\frac{c(\mathbf{x})}{c(\mathbf{x})-\Gamma_{\mathrm{Stg}}(\mathbf{x})},0\right)<\eta\leq 1\}$. Therefore, Theorem~\ref{Property_for_U_with_Kappa} suggests that we should identify a smooth function $\eta$ such that $\eta \in \Xi$, and then construct the tunable term $\kappa(\mathbf{x})$ using the expression $\kappa(\mathbf{x})=(1-\eta)\cdot\frac{c(\mathbf{x})}{\Gamma_{\mathrm{Stg}}(\mathbf{x})}+\eta$ to construct the tunable term $\kappa(\mathbf{x})$. For example, choosing $\eta = \frac{1}{2}$ or $\eta = 1$ yields $\kappa(\mathbf{x}) = \frac{c(\mathbf{x}) + \Gamma_{\mathrm{Stg}}(\mathbf{x})}{2\Gamma_{\mathrm{Stg}}(\mathbf{x})}$ and $\kappa(\mathbf{x}) = 1$, which correspond to the Half-Sontag formula~\cite{cohen2023characterizing} and Sontag’s formula~\cite{Sontag_fromula}, respectively.
\end{Rmk}

Although Remark~\ref{eta_range} provides the valid range for $\eta$ and offers some examples, selecting a valid $\eta$ is not an easy task due to its state dependence. The following theorem provides a state-independent (yet more conservative) range for $\eta$, which ensures both smoothness and safety simultaneously.
\begin{Thm}\label{Eta_Constant_Choice}
Let $h:\mathbb{R}^{n}\to\mathbb{R}$ be a CBF, and let
$\eta:\mathbb{R}^{2}\to\mathbb{R}$ be smooth on the domain $\Phi$ with
$0.5\leq\eta(\mathbf{x})\leq 1$ for every $\mathbf{x}\in\Phi$.  Use $\kappa(\mathbf{x})
    \;=\;
    \bigl(1-\eta\bigr)\cdot
    \frac{c(\mathbf{x})}{\Gamma_{\mathrm{Stg}}(\mathbf{x})}
    +\eta$ to derive the tunable universal formula~\eqref{Tunable_Universal_Formula}. Then the resulting control law is smooth and the closed-loop system~\eqref{Affine_Control_System} is safe for all states $\mathbf{x}\in\mathbb{R}^{n}$.
\end{Thm}
\begin{proof}
As stated in Remark~\ref{eta_range}, if $\eta\in\Xi=\{\eta:\Phi\rightarrow\mathbb{R}|\max(\frac{c(\mathbf{x})}{c(\mathbf{x})-\Gamma_{\mathrm{Stg}}(\mathbf{x})},0)<\eta\leq 1\}$, then we can obtain $\kappa(\mathbf{x})\in\mathcal{K}_{\mathrm{SM}}$. Due to that $0.5\leq \eta(\mathbf{x})\leq 1 $ always implies $\eta \in \Xi$, Theorem~\ref{Property_for_U_with_Kappa} ensures that the resulting control law is smooth and that the closed‑loop system~\eqref{Affine_Control_System} remains safe for every state $\mathbf{x}\in\mathbb{R}^{n}$.
\end{proof}

In the following proposition, we will demonstrate that the tunable universal formula is never overly conservative (like Sontag's formula), as it can achieve an arbitrarily close approximation to the QP-synthesized controller through parameter tuning.
\begin{Prop}\label{Perfect_Approximation}
Let $h:\mathbb{R}^{n}\to\mathbb{R}$ be a CBF.  
Fix the smooth function $\eta$ to be a constant value $\eta\;=\; 0.5$.
Define $\kappa(\mathbf{x})
  \;=\;
  \bigl(1-\eta\bigr)\,
  \frac{c(\mathbf{x})}{\Gamma_{\mathrm{Stg}}(\mathbf{x})}
  \;+\;
  \eta$ and construct the tunable universal formula $\mathbf{k}_{\text {Tun }}(\mathbf{x})$ as in~\eqref{Tunable_Universal_Formula}. If the function \(s:\mathbb{R}\!\to\!\mathbb{R}\) in \eqref{Gamma_Choice} is chosen as \(s(d)=\sigma d\), with \(\sigma\to 0\), then one can obtain $\mathbf{k}_{\mathrm{Tun}}(\mathbf{x}) \;\rightarrow\; \mathbf{k}_{\mathrm{QP}}(\mathbf{x})$.
\end{Prop}
\begin{proof}
Under the condition $\eta = 0.5$, we will obtain $\kappa(\mathbf{x}) = \frac{c(\mathbf{x}) + \Gamma_{\mathrm{Stg}}(\mathbf{x})}{2\Gamma_{\mathrm{Stg}}(\mathbf{x})}$ according to its definition. Then, substituting $\kappa(\mathbf{x})$ into the tunable universal formula $\mathbf{k}_{\mathrm{Tun}}(\mathbf{x})$ as defined by~\eqref{Tunable_Universal_Formula}, we obtain $\mathbf{k}_{\mathrm{Tun}}(\mathbf{x}) = \frac{1}{2} \mathbf{k}_{\mathrm{Stg}}(\mathbf{x})$. Next, we examine the definition of $\mathbf{k}_{\mathrm{Stg}}(\mathbf{x})$ provided in~\eqref{CBF_Universal_Law} and select $s(d) = \sigma d$ with $\sigma \to 0$. This choice simplifies the expression for $\mathbf{k}_{\mathrm{Stg}}(\mathbf{x})$, and it immediately gives us $\mathbf{k}_{\mathrm{Tun}}(\mathbf{x}) =  \frac{1}{2} \mathbf{k}_{\mathrm{Stg}}(\mathbf{x})\rightarrow\mathbf{k}_{\mathrm{QP}}(\mathbf{x})$. This result shows that, under the choice of $\eta = 0.5$ and $s(d) = \sigma d$ with $\sigma \rightarrow 0$, the tunable universal formula can approach the QP-synthesized controller arbitrarily closely.
\end{proof}
\begin{Rmk}
    Note that, in Proposition~\ref{Perfect_Approximation}, the Squareplus approximation~\cite{barron2021squareplus} under the joint limit of $\eta = 0.5$ and $s(d) = \sigma d$ with $\sigma \to 0$ is the primary reason why the tunable universal formula can perfectly approximate the QP-synthesized controller. Specifically, we start by setting \(\eta=0.5\) in the tunable universal formula $\mathbf{k}_{\mathrm{Tun}}(\mathbf{x})$, which reduces it to Half-Sontag’s formula (cf. Remark~\ref{eta_range}). Next, we introduce a one-parameter smooth ``Squareplus" approximation by defining $s(d) = \sigma d$ with $\sigma \to 0$. This choice replaces the sharp ReLU function in $\lambda_{\mathrm{PMN}}(c,d)$ (cf. Eq.\eqref{Relu_Function}) with the smoothed Half-Sontag–based function $\lambda_{\mathrm{Stg}}(c, d)$ (cf. Eq.\eqref{Stg_Function}). Consequently, we obtain $\mathbf{k}_{\mathrm{Tun}}(\mathbf{x}) = \frac{1}{2} \mathbf{k}_{\mathrm{Stg}}(\mathbf{x}) \rightarrow\mathbf{k}_{\mathrm{QP}}(\mathbf{x})$.
\end{Rmk}

\subsubsection{Robustness} 
While the term ``safety margin" is commonly used in existing literature, a precise definition has not been established. In~\cite{Safety_Margin}, the concept of stability margin~\cite{Universal_Formula_Gain_Margin_1} has been extended to safety-critical control. This extension still allows for a quantitative assessment of the robustness of a controller, similar to that of a stabilizing control. Therefore, we follow the definition of stability margin in~\cite{Stability_Margin} and adapt it to the definition of safety margin.
\begin{Def}\label{Stability_Margin}
    (Safety margin) A safety control law, $\mathbf{u}=\mathbf{k}(\mathbf{x})$, has safety margins $(m_{1},m_{2})$,
    \begin{equation*}
-1\leq m_1<m_2 \leq \infty,
\end{equation*}
if, for every constant $\xi \in [m_{1},m_{2})$, the control $\bar{\mathbf{u}}=(1+\xi)\mathbf{k}(\mathbf{x})$ also ensures the safety of the system.
\end{Def}
\begin{Thm}\label{Gain_Margin_Theorem}
The safety margin of the tunable universal formula presented in~\eqref{Tunable_Universal_Formula} is $[\bar{\xi},\infty)$, where $\bar{\xi}:=\underset{\mathbf{x}\in\mathbb{R}^{n}}{\sup}\mathcal{M}(\mathbf{x})$,
\begin{equation}\label{Bound_Function}
    \mathcal{M}(\mathbf{x})\triangleq -1+\frac{c(\mathbf{x})}{c(\mathbf{x})-\kappa(\mathbf{x})\Gamma_{\mathrm{Stg}}(\mathbf{x})},
\end{equation}
and $\kappa(\mathbf{x})\in\mathcal{K}_{\mathrm{SM}}$.
\end{Thm}
\begin{proof}
Given that the tunable universal formula defined in~\eqref{Tunable_Universal_Formula} satisfies the tightened CBF condition presented in~\eqref{Scaling_Condition_kappa} with $\kappa(\mathbf{x})\in\mathcal{K}_{\mathrm{SM}}$, we substitute $\mathbf{k}_{\mathrm{Tun}}(\mathbf{x})$ into~\eqref{Scaling_Condition_kappa} and add $\xi\mathbf{d}(\mathbf{x})\mathbf{k}_{\mathrm{Tun}}(\mathbf{x})$ to both sides of~\eqref{Scaling_Condition_kappa} based on the safety margin definition. This yields:
    \begin{equation}\label{Stability_margin_proof}
        \begin{split}
            &c(\mathbf{x})+(1+\xi)\mathbf{d}(\mathbf{x}) \mathbf{k}_{\mathrm{Tun}}(\mathbf{x})\\
            &\qquad\qquad\qquad\geq \xi\mathbf{d}(\mathbf{x}) \mathbf{k}_{\mathrm{Tun}}(\mathbf{x})+\kappa(\mathbf{x})\Gamma_{\mathrm{Stg}}(\mathbf{x}).
        \end{split}
    \end{equation}
    
A sufficient condition for guaranteeing the safety of the system~\eqref{Affine_Control_System} is that the right-hand side of~\eqref{Stability_margin_proof} is non-negative for all $\mathbf{x}\in\mathbb{R}^{n}$, which leads to $\xi \geq \mathcal{M}(\mathbf{x})$.
According to the definition $\bar{\xi}:=\sup_{\mathbf{x}\in\mathbb{R}^{n}}\mathcal{M}(\mathbf{x})$, we know that $\mathbf{k}_{\mathrm{Tun}}(\mathbf{x})$ has a safety margin $[m_{1},m_{2})$ with $m_{1}=\bar{\xi}$ and $m_{2}=\infty$.
\end{proof}
\begin{Rmk}
The condition $\kappa(\mathbf{x})\in\mathcal{K}_{\mathrm{SM}}$ ensures that the function $\mathcal{M}(\mathbf{x})= 0$ in~\eqref{Bound_Function} is always satisfied. Thereby, it implies that there always exists a constant $\bar{\xi}=\sup_{\mathbf{x}\in\mathbb{R}^{n}}\mathcal{M}(\mathbf{x})= 0$. Further, we emphasize that the safety margin of the tunable universal formula can be more precisely determined with a specific $\kappa(\mathbf{x})$. For example, setting $\kappa(\mathbf{x})=1$ corresponds to Sontag's formula, leading to a safety margin $\bar{\xi}\in[-\frac{1}{2},\infty)$ due to the relationship $\bar{\xi}=\sup_{\mathbf{x}\in\mathbb{R}^{n}}\mathcal{M}(\mathbf{x})=\sup_{\mathbf{x}\in\mathbb{R}^{n}} \left(\frac{\Gamma_{\mathrm{Stg}}(\mathbf{x})}{c(\mathbf{x})-\Gamma_{\mathrm{Stg}}(\mathbf{x})}\right)=-\frac{1}{2}, \forall\mathbf{x}\in\mathbb{R}^{n}$. This conclusion aligns with the established condition for the stability margin of Sontag's formula, which can be found in~\cite{Safety_Margin}.
\end{Rmk}

\subsubsection{Discussions}
In Theorem~\ref{Safety_Guarantee_Thm}, we have established the conditions $\kappa(\mathbf{x})\in\mathcal{K}_{\mathrm{SA}}$ for the tunable universal formula $\bar{\mathbf{k}}_{\mathrm{Tun}}(\mathbf{x})$ to ensure the safety of the closed-loop system~\eqref{Affine_Control_System}, where $\bar{\mathbf{k}}_{\mathrm{Tun}}(\mathbf{x})$ is not necessarily assumed to be smooth. Then, we demonstrate that a family of smooth controllers (with safety guarantees) can be designed according to Theorem~\ref{Property_for_U_with_Kappa} with different $\kappa(\mathbf{x})\in\mathcal{K}_{\mathrm{SM}}$ (relating to different $\eta$). This suggests that both safety and smoothness properties can be achieved simultaneously by appropriately selecting $\kappa(\mathbf{x})$. However, comparing~\eqref{Kappa_SA} and~\eqref{Kappa_SM} shows that $\mathcal{K}_{\mathrm{SM}} \subset \mathcal{K}_{\mathrm{SA}}$, because $\mathcal{K}_{\mathrm{SM}}$ restricts the admissible choices of $\kappa(\mathbf{x})$ to satisfy both safety and smoothness. Proposition~\ref{Perfect_Approximation} shows that the tunable universal formula exactly reproduces the QP-synthesized controller when \(\eta = 0.5\) and \(\sigma = 0\), and, as noted in Remark~\ref{eta_range}, setting \(\eta = 1\) with \(\sigma \neq 0\) recovers Sontag’s formula. Therefore, the tunable universal formula unifies and generalizes both the QP-synthesized controller and Sontag’s formula.  Moreover, Theorems~\ref{Property_for_U_with_Kappa} and~\ref{Gain_Margin_Theorem} show that, by appropriately tuning \(\eta\) and \(\sigma\), the tunable universal formula improves the robustness and smoothness of the QP-based controller while simultaneously reducing the conservatism of Sontag’s formula.
Additionally, to evaluate the robustness of these safe and smooth controllers, one can refer to the conclusions in Theorem~\ref{Gain_Margin_Theorem}. As noticed in~\eqref{Bound_Function}, an increase in $\kappa(\mathbf{x})$ will lead to a large $\mathcal{M}(\mathbf{x})$, resulting in a larger safety margin, and hence the control law is more robust. However, a large $\kappa(\mathbf{x})$ will result in a more conservative safety-critical control because of~\eqref{Scaling_Condition_kappa} and by following the discussions in Remark~\ref{Rmk_Conservatism}. Therefore, there is an inherent trade-off between the robustness of a safe and smooth control law. In practice, the choice of a specific $\kappa(\mathbf{x})$, which leads to certain levels of robustness, relies on the specific demands of the application.
\subsection{Addressing Norm-Bounded Input Constraints}
In this subsection, we aim to show that the safety margin of the tunable universal formula~\eqref{Tunable_Universal_Formula} can be utilized to manage input constraints. Similar discussions were presented in our previous work~\cite{li2024unifying}, which focused on generalizing Lin-Sontag's formula for stabilizing control under norm-bounded input constraints. In this paper, we extend the analysis to explore how a general tunable universal formula can address norm-bounded input constraints in safety-critical control. While the techniques are similar, the insights into stabilizing control and safety-critical control differ significantly, and we minimize overlap between the two papers in what follows.

We consider the following control input constraint.
\begin{equation}\label{Control_Input_Limit}
    \|\mathbf{u}\|\leq \gamma,
\end{equation}
where $\gamma> 0$ is a constant. 
\begin{Def}\label{Compatibility_Def}
($\mathbf{x}$-Compatibility) The CBF condition is $\mathbf{x}$-compatible with the norm-bounded input constraint at $\mathbf{x}\in\mathbb{R}^{n}$ if there exists a $\mathbf{u}\in\mathbb{R}^{m}$ that satisfies~\eqref{CBF_Condition} and~\eqref{Control_Input_Limit} simultaneously. We say that the CBF and norm-bounded input constraints are compatible if the CBF is $\mathbf{x}$-compatible with the norm-bounded input constraint~\eqref{Control_Input_Limit} $\forall\mathbf{x} \in \mathcal{C}$.
\end{Def}
\begin{Asum}\label{Compatibility_Assum}
The CBF is compatible with the norm-bounded input constraint.
\end{Asum}

To ensure Assumption~\ref{Compatibility_Assum} is satisfied, it is essentially equivalent to satisfying the conditions in the following lemma.
\begin{Lem}\label{Compa_Lem}
    The CBF condition given by~\eqref{CBF_Condition} is compatible with the input constraint~\eqref{Control_Input_Limit} if and only if there exists $\mathbf{x}\in\mathcal{S}_{\mathrm{cpt}}:=\{\mathbf{x}\in\mathbb{R}^{n}:\gamma\|\mathbf{d}(\mathbf{x})\|\geq -c(\mathbf{x})\}$.
\end{Lem}
\begin{proof}
Suppose there exists a state-feedback control law $\mathbf{k}(\mathbf{x})$ satisfying the CBF condition~\eqref{CBF_Condition}. To ensure that the CBF condition~\eqref{CBF_Condition} is $\mathbf{x}$-compatible with the norm-bounded input constraint, a necessary and sufficient condition is that there always exists a control law $\mathbf{k}(\mathbf{x})$ such that $\|\mathbf{k}(\mathbf{x})\|\leq\gamma$. In other words, we need guarantee that the norm of the control law $\mathbf{k}(\mathbf{x}):=\mathrm{argmin}_{\mathbf{u}}\|\mathbf{u}\|,\mathrm{s.t.} c(\mathbf{x})+\mathbf{d}(\mathbf{x})\mathbf{u}\geq 0$ is less than $\gamma$. Notably, this is equivalent to requiring $\|\mathbf{k}_{\mathrm{QP}}(\mathbf{x})\|\leq\gamma$.  By substituting~\eqref{QP_Synthesized} into the inequality $\|\mathbf{k}_{\mathrm{QP}}(\mathbf{x})\|\leq\gamma$, the condition $\gamma\|\mathbf{d}(\mathbf{x})\|\geq -c(\mathbf{x})$ is obtained.
\end{proof}
\subsubsection{Safety}\label{Safety_Input_Constraint}
When applying the tunable universal formula $\mathbf{k}_{\mathrm{Tun}}(\mathbf{x})$ from~\eqref{Tunable_Universal_Formula} to safety-critical control with input constraints, the choices for selecting $\kappa(\mathbf{x})$ become more limited compared to the case without input constraints in Section~\ref{Tun_Without_Input}. This is because $\mathbf{k}_{\mathrm{Tun}}(\mathbf{x})$ must satisfy both the CBF condition and the input constraint simultaneously.

\begin{Thm}\label{Scaling_Universal_Formula}
Assume that $h:\mathbb{R}^{n}\rightarrow\mathbb{R}$ is a CBF and Assumption~\ref{Compatibility_Assum} holds. The following tunable universal formula law ensures the closed-loop system~\eqref{Affine_Control_System} is safe and satisfies the norm-bounded input constraint~\eqref{Control_Input_Limit} simultaneously:
\begin{equation}\label{QP_CLF_Solution_with_tuning}
    \begin{split}
    \widetilde{\mathbf{k}}_{\mathrm{Tun-BI}}(\mathbf{x})=\bar{\lambda}_{\mathrm{Tun}}(c(\mathbf{x}),\|\mathbf{d}( \mathbf{x}) \|^{2},\widetilde{\kappa}(\mathbf{x})) \mathbf{d}(\mathbf{x})^{\top}
    \end{split}
\end{equation}
where $\widetilde{\kappa}(\mathbf{x})\in\widetilde{\mathcal{K}}_{\mathrm{SA-BI}}$, $\widetilde{\mathcal{K}}_{\mathrm{SA-BI}}=\{\widetilde{\kappa}:\mathbb{R}^{n}\rightarrow\mathbb{R} |0<\widetilde{\kappa}(\mathbf{x})\leq\frac{\gamma\|\mathbf{d}(\mathbf{x})\|+c(\mathbf{x})}{\Gamma_{\mathrm{Stg}}(\mathbf{x})}\}$.
\end{Thm}
\begin{proof}
Due to Assumption~\ref{Compatibility_Assum}, we have $\gamma\|\mathbf{d}(\mathbf{x})\|\geq -c(\mathbf{x})$ using Lemma~\ref{Compa_Lem}. Therefore, we can always ensure that $\widetilde{\mathcal{K}}_{\mathrm{SA-BI}}\neq\emptyset$. Next, we examine whether $\widetilde{\mathbf{k}}_{\mathrm{Tun-BI}}(\mathbf{x})$ satisfies the CBF condition~\eqref{CBF_Condition}. Firstly, substituting $\mathbf{d}(\mathbf{x})=\mathbf{0}$ into~\eqref{QP_CLF_Solution_with_tuning}, it gives $\widetilde{\mathbf{k}}_{\mathrm{Tun-BI}}(\mathbf{x})=\mathbf{0}$. Given that $c(\mathbf{x})>0$ when $\mathbf{d}(\mathbf{x})=\mathbf{0}$, then it can be verified that the CBF condition is satisfied by substituting $\widetilde{\mathbf{k}}_{\mathrm{Tun-BI}}(\mathbf{x})=\mathbf{0}$ into~\eqref{CBF_Condition}. For the case  $\mathbf{d}(\mathbf{x})\neq\mathbf{0}$, we know that $\widetilde{\mathbf{k}}_{\mathrm{Tun-BI}}(\mathbf{x})=\frac{\widetilde{\kappa}(\mathbf{x})\Gamma_{\mathrm{Stg}}(\mathbf{x})-c(\mathbf{x})}{\|\mathbf{d}(\mathbf{x})\|^{2}} \mathbf{d}(\mathbf{x})^{\top}$. Substituting $\widetilde{\mathbf{k}}_{\mathrm{Tun-BI}}(\mathbf{x})$ into~\eqref{CBF_Condition} gives 
\begin{equation}\nonumber
    c(\mathbf{x})+\mathbf{d}(\mathbf{x}) \cdot \frac{\widetilde{\kappa}(\mathbf{x})\Gamma_{\mathrm{Stg}}(\mathbf{x})-c(\mathbf{x})}{\|\mathbf{d}(\mathbf{x})\|^{2}} \mathbf{d}(\mathbf{x})^{\top}=\widetilde{\kappa}(\mathbf{x})\Gamma_{\mathrm{Stg}}(\mathbf{x})\geq 0.
\end{equation}
Finally, we need to verify that $\widetilde{\mathbf{k}}_{\mathrm{Tun-BI}}(\mathbf{x})$ remains within the specified control input range defined by $\|\mathbf{u}\|\leq\gamma$. When $\mathbf{d}(\mathbf{x})=\mathbf{0}$, the control law $\widetilde{\mathbf{k}}_{\mathrm{Tun-BI}}(\mathbf{x})=\mathbf{0}$ always satisfies the input constraint. For the case $\mathbf{d}(\mathbf{x})\neq\mathbf{0}$, we notice that    
\begin{equation}\nonumber
   \|\widetilde{\mathbf{k}}_{\mathrm{Tun-BI}}(\mathbf{x})\|=\left\|\frac{-c(\mathbf{x})+\widetilde{\kappa}(\mathbf{x})\Gamma_{\mathrm{Stg}}(\mathbf{x})}{\|\mathbf{d}(\mathbf{x})\|^{2}} \mathbf{d}(\mathbf{x})^{\top}\right\|\leq\gamma
\end{equation}
given the condition $\widetilde{\kappa}(\mathbf{x})\in\widetilde{\mathcal{K}}_{\mathrm{SA-BI}}$. Therefore, we conclude that the closed-loop system~\eqref{Affine_Control_System} is safe and simultaneously satisfies the input constraint~\eqref{Control_Input_Limit}.
\end{proof}

\subsubsection{Smoothness}
In the following theorem, we provide conditions to guarantee the controller defined by~\eqref{QP_CLF_Solution_with_tuning} is smooth.
\begin{Thm}\label{Smoothness}
Suppose that $h:\mathbb{R}^{n}\rightarrow\mathbb{R}$ is a CBF, Assumption~\ref{Compatibility_Assum} holds, and $\widetilde{\eta}:\mathbb{R}^{2}\rightarrow\mathbb{R}$ is a smooth function on $\Phi$. For the tunable universal formula provided in~\eqref{QP_CLF_Solution_with_tuning}, we choose $\widetilde{\kappa}(\mathbf{x})=(1-\widetilde{\eta})\cdot\frac{c(\mathbf{x})}{\Gamma_{\mathrm{Stg}}(\mathbf{x})}+\widetilde{\eta}$ and ensure $\widetilde{\kappa}(\mathbf{x})\in\widetilde{\mathcal{K}}_{\mathrm{SM-BI}}$, where $\widetilde{\mathcal{K}}_{\mathrm{SM-BI}}:=\left\{\widetilde{\kappa}:\mathbb{R}^{n}\rightarrow\mathbb{R}|\max\left(\frac{c(\mathbf{x})}{\Gamma_{\mathrm{Stg}}(\mathbf{x})},0\right)< \widetilde{\kappa}(\mathbf{x})\leq\frac{\gamma\|\mathbf{d}(\mathbf{x})\|+c(\mathbf{x})}{\Gamma_{\mathrm{Stg}}(\mathbf{x})}\right\}$. Then the resulting control law is smooth and ensures the closed-loop system~\eqref{Affine_Control_System} is safe for all $\mathbf{x}\in\mathbb{R}^{n}$.
\end{Thm}
\begin{proof}
One can follow the proof given in Theorem~\ref{Property_for_U_with_Kappa} to verify the smoothness of $\widetilde{\mathbf{k}}_{\mathrm{Tun-BI}}(\mathbf{x})$.
\end{proof}
\begin{Rmk}
    Note that ensuring $\widetilde{\kappa}(\mathbf{x}) \in \widetilde{\mathcal{K}}_{\mathrm{SM-BI}}$ requires constructing a smooth function $\widetilde{\eta}$ such that $\widetilde{\eta} \in \widetilde{\Xi}$, where $\widetilde{\Xi}:=\left\{\widetilde{\eta}:\Phi\rightarrow\mathbb{R}|\max\left(\frac{c(\mathbf{x})}{c(\mathbf{x})-\Gamma_{\mathrm{Stg}}(\mathbf{x})},0\right)<\widetilde{\eta}\leq \frac{\gamma\|\mathbf{d}(\mathbf{x})\|}{\Gamma_{\mathrm{Stg}}(\mathbf{x})-c(\mathbf{x})}\right\}$. A feasible choice for $\widetilde{\eta}$ is $\widetilde{\eta} = \frac{1}{\sqrt{s(\|\mathbf{d}(\mathbf{x}) \|^{2})/\gamma^2+1}+1}$. With this choice, the tunable universal formula reduces to the Lin-Sontag formula presented in~\cite{Norm_Bounded} for safety-critical control. Alternative feasible choice can also be found, such as $\tilde{\eta}(\mathbf{x}):=\frac{\gamma\|\mathbf{d}(\mathbf{x})\|-c(\mathbf{x})}{2\left(\Gamma_{\mathrm{Stg}}(\mathbf{x})-c(\mathbf{x})\right)}$.
\end{Rmk}

\subsubsection{Robustness} 
We continue to use Definition~\ref{Stability_Margin} and Theorem~\ref{Gain_Margin_Theorem} to assess the safety margin associated with the tunable universal formula given in~\eqref{QP_CLF_Solution_with_tuning}. However, for~\eqref{QP_CLF_Solution_with_tuning}, we rely on $\widetilde{\kappa}(\mathbf{x})$ for defining $\mathcal{M}(\mathbf{x})$ (see Equation~\eqref{Bound_Function}). With the following theorem, we determine that $\bar{\xi}=0$. 
\begin{Thm}\label{Gain_Margin_Theorem_Input}
    The controller defined by the tunable universal formula given in~\eqref{QP_CLF_Solution_with_tuning} has a safety margin $\bar{\xi}\in[0,\infty)$.
\end{Thm}
\begin{proof}
As $\widetilde{\mathcal{K}}_{\mathrm{SA-BI}}\subseteq\mathcal{K}_{\mathrm{SA}}$, Theorem~\ref{Gain_Margin_Theorem} remains applicable to~\eqref{QP_CLF_Solution_with_tuning}. This implies that the safety margin for~\eqref{QP_CLF_Solution_with_tuning} is within $\xi\in[\bar{\xi},\infty)$, $\bar{\xi}:=\sup_{\mathbf{x}\in\mathbb{R}^{n}}\mathcal{M}(\mathbf{x})$. Given that $\widetilde{\kappa}(\mathbf{x})\in\widetilde{\mathcal{K}}_{\mathrm{SA-BI}}$, we can deduce that $\bar{\xi}= 0$ based on~\eqref{Bound_Function}, which indicates that the safety margin of $\widetilde{\mathbf{k}}_{\mathrm{Tun-BI}}(\mathbf{x})$ is $[0,\infty)$.
\end{proof}
\subsubsection{Discussions}\label{Norm_Discussions}
As demonstrated in Theorem~\ref{Scaling_Universal_Formula}, one can derive a tunable universal formula~\eqref{QP_CLF_Solution_with_tuning} (with a proper $\widetilde{\kappa}(\mathbf{x})$) to address the safety-critical control problem with a norm-bounded input constraint. This essentially involves leveraging the safety margin defined in Definition~\ref{Stability_Margin} to accommodate input constraints. In contrast to Theorem~\ref{Safety_Guarantee_Thm}, the only difference is the range of the tunable term. 
Given that the input constraint is accounted for by the tunable universal formula~\eqref{QP_CLF_Solution_with_tuning}, the range of the tunable term becomes more limited. Moreover, Theorem~\ref{Smoothness} shows that the design of a smooth tunable universal formula remains unaffected, compared to Theorem~\ref{Property_for_U_with_Kappa}, as long as it stays within the range of the tunable term Finally, Theorem~\ref{Gain_Margin_Theorem_Input} demonstrates that, with an adjustment to the validity range of the tunable term (noting the reduced safety margin in~\eqref{QP_CLF_Solution_with_tuning} compared to~\eqref{Tunable_Universal_Formula}), we can still regulate the controller's performance in terms of safety, smoothness, and robustness, while satisfying a norm-bounded input constraint.
\section{Application Studies}\label{Simulation}
In this section, we demonstrate the advantages of the proposed tunable universal formula for safety-critical control, balancing safety, smoothness, and robustness, using a two-link robot manipulator~\cite{soltanpour2009robust}. We further show that its inherent robustness enables the controller to address input constraints. 

\begin{figure}[tp]\label{Comparison}
 \centering
    \makebox[0pt]{%
    \includegraphics[width=0.775\columnwidth]{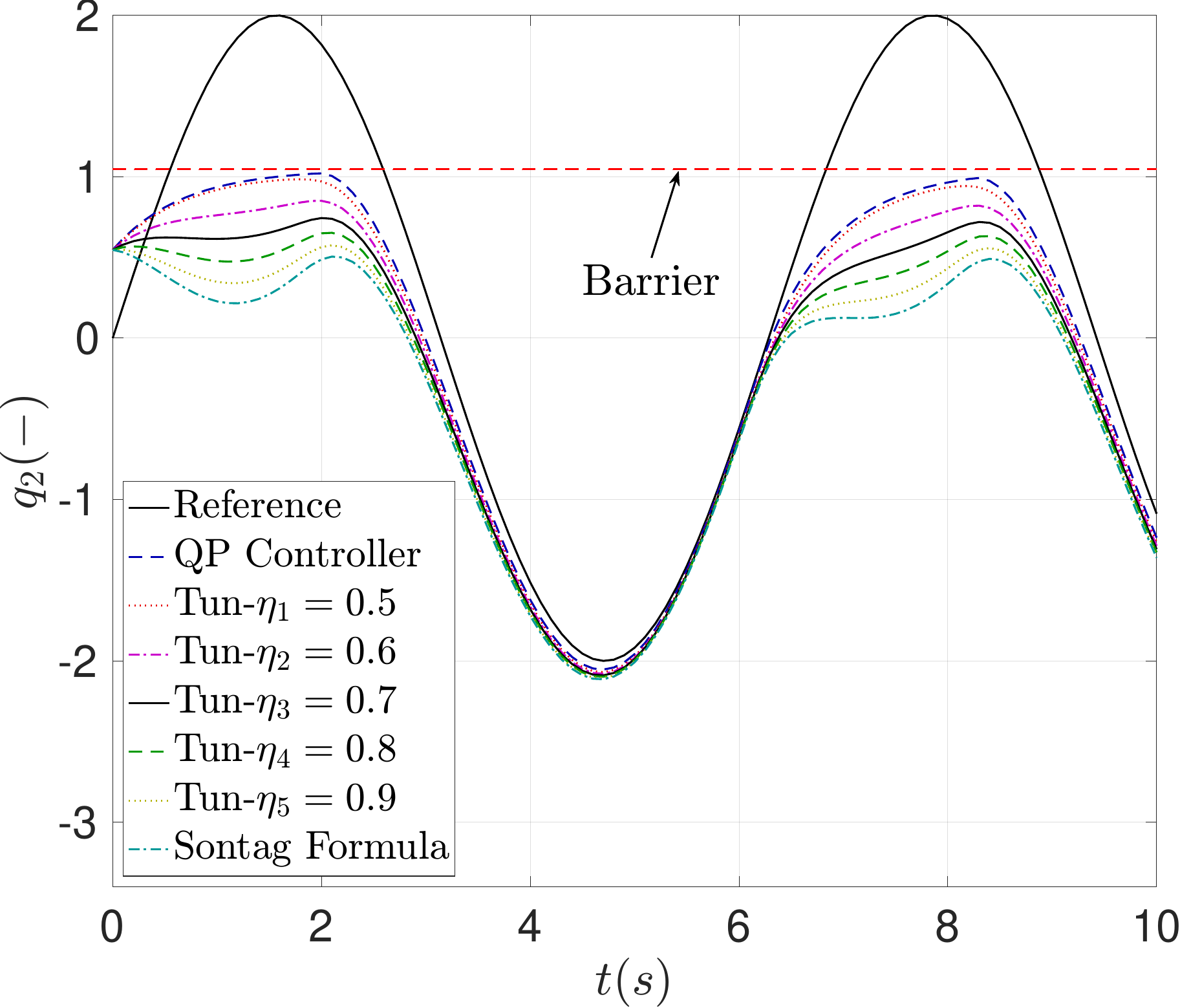}}
    \caption{The trajectories of the system states under the virtual controller $\mathbf{v} = \mathbf{k}_{0}$ are shown, which utilize tunable universal formulas with different tunable terms, with a comparison to both a QP-synthesized controller and Sontag's formula.}
    \label{Trajectory}
\end{figure}
\begin{figure}[tp]\label{Comparison}
 \centering
    \makebox[0pt]{%
    \includegraphics[width=0.775\columnwidth]{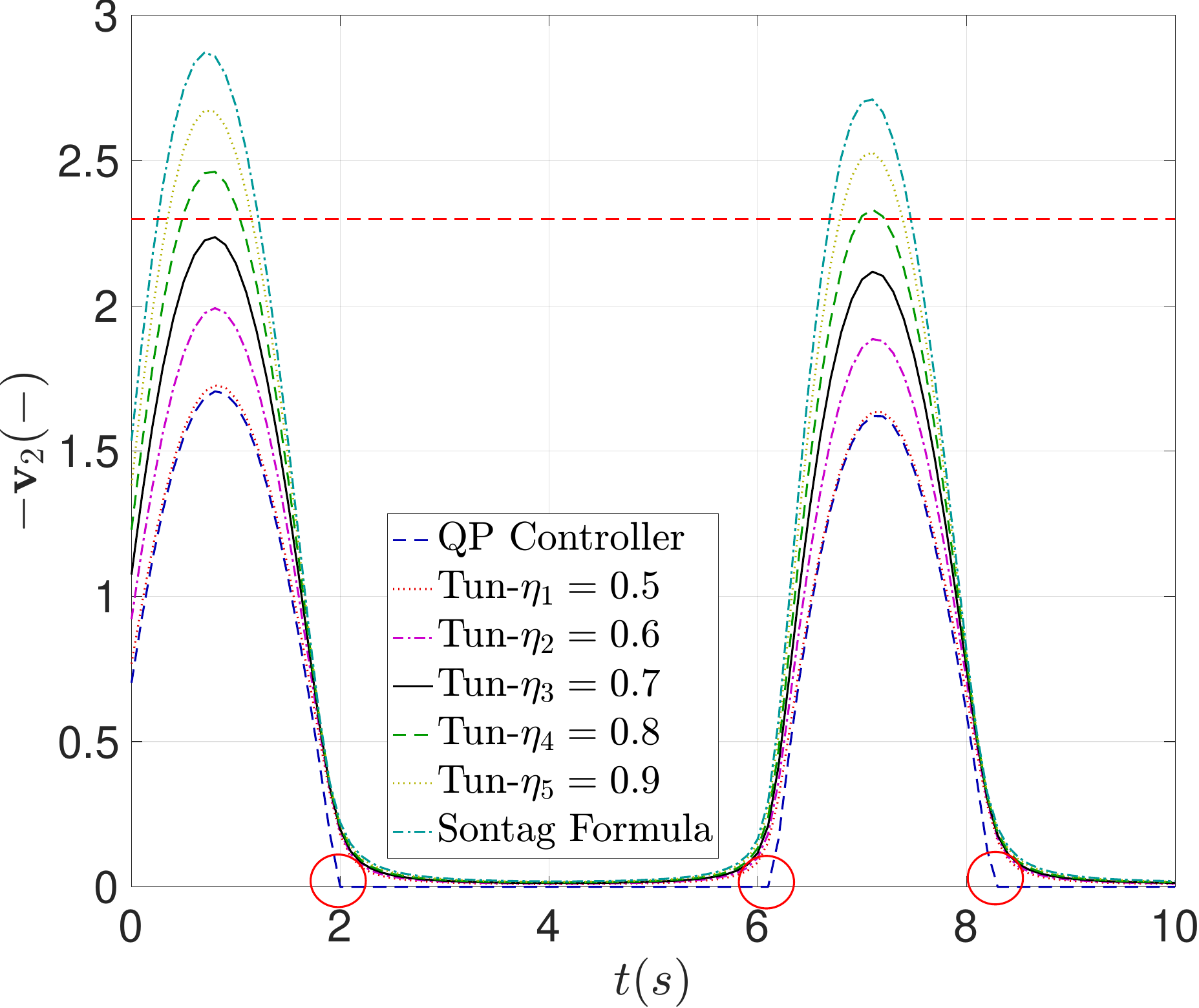}}
    \caption{The control input (with sign inversion), $-\mathbf{v}_{2}$, generated by the virtual controller $\mathbf{v} = \mathbf{k}_{0}$ using tunable universal formulas with various tunable terms, is compared to both the QP-synthesized controller and Sontag's formula.}
    \label{Input_Behavior}
\end{figure}
The two-link manipulator is modeled as follows:
\begin{equation}\label{Segway_model}
   \mathbf{M}(\mathbf{q})\ddot{\mathbf{q}}+\mathbf{C}(\mathbf{q},\dot{\mathbf{q}})\dot{\mathbf{q}}+\mathbf{N}(\mathbf{q})=\mathbf{u}, 
\end{equation}
where $\mathbf{q}=[q_{1},q_{2}]^{\top}\in\mathbb{R}^{2}$ describes the joint angle, the details for matrices $\mathbf{M}(\mathbf{q}), \mathbf{C}(\mathbf{q},\dot{\mathbf{q}}), \mathbf{N}(\mathbf{q})$ are provided in~\cite{soltanpour2009robust}, and we assume that $\mathbf{M}(\mathbf{q})$ is invertible. Consequently, we reformulate~\eqref{Segway_model} as:
\begin{equation}\label{Segway_reformulation}
\begin{split}
    \dot{\mathbf{q}}=\mathbf{v}, \quad
    \dot{\mathbf{v}}=\Phi(\mathbf{q},\dot{\mathbf{q}})+\mathbf{H}(\mathbf{q})\mathbf{u},
\end{split}
\end{equation}
where $\Phi(\mathbf{q},\dot{\mathbf{q}})=-\mathbf{M}(\mathbf{q})^{-1}\left(\mathbf{C}(\mathbf{q},\dot{\mathbf{q}})\dot{\mathbf{q}}+\mathbf{N}(\mathbf{q})\right)$, $\mathbf{H}(\mathbf{q})=\mathbf{M}(\mathbf{q})^{-1}$, and we define $\mathbf{x}:=[\mathbf{q}^{\top},\mathbf{v}^{\top}]^{\top}$. The objective is to control the two-link manipulator to track the desired trajectory $\mathbf{q}_{\mathrm{d}}=[2\sin(t)+1,2\sin(t)]^{\top}$ and to ensure that the manipulator remains within the safe region defined by $\mathcal{C}=\{\mathbf{x} \in \mathbb{R}^4: h(\mathbf{x})= \bar{q}_{2}-q_{2}\geq 0\}$ with $\bar{q}_{2}=\frac{\pi}{3}$. 

Given that~\eqref{Segway_reformulation} is a second-order system, one can apply the safe backstepping algorithms from~\cite{taylor2022safe} to design a safe controller. This approach involves first designing a virtual safety controller \(\mathbf{v}=\mathbf{k}_0(\mathbf{x})\) for the system \(\dot{\mathbf{q}} = \mathbf{v}\) by assuming that we could direct control over \(\mathbf{v}\).  Then, a barrier function $\mathfrak{b}(\mathbf{q},\mathbf{v})=h(\mathbf{x})-\frac{1}{2\mu}(\mathbf{v}-\mathbf{k}_0(\mathbf{x}))^{\top}(\mathbf{v}-\mathbf{k}_0(\mathbf{x}))$ is constructed to derive the control input \(\mathbf{u}\) that ensures the safety of the system~\eqref{Segway_reformulation}. However, as noted in~\cite{taylor2022safe}, the safe backstepping algorithm requires \(\mathbf{k}_0(\mathbf{x})\) to be smooth. This requirement limits the direct use of the QP-synthesized controller for designing \(\mathbf{k}_0(\mathbf{x})\) due to its lack of smoothness at certain points (see Section~\ref{Smoothness_exist}). Alternatively, one can follow the approach discussed in~\cite{molnar2021model}. In particular, once $\mathbf{k}_{0}(\mathbf{x})$ is designed, a PD controller can be further designed to track the safe trajectory. However, this approach still requires $\mathbf{k}_{0}(\mathbf{x})$ to be smooth~\cite{cohen2023characterizing}, which prevents the use of the QP-synthesized controller.

In this paper, we use the safe backstepping algorithms from~\cite{taylor2022safe} to address safety-critical control, where we start with the design of the virtual controller \(\mathbf{v}=\mathbf{k}_0(\mathbf{x})\) using QP-synthesized controller, tunable universal formulas, and Sontag's formula. In particular, we first introduce a nominal stabilizing control input $\mathbf{k}_{0,\mathrm{d}}(\mathbf{x})=-\mathbf{K}_{\rm{P}}\mathbf{e}+\dot{\mathbf{q}}_{\rm{d}}$ to achieve the tracking task, where $\mathbf{e}=\mathbf{q}-\mathbf{q}_{\rm{d}}$, $\mathbf{K}_{\rm{P}}=\rm{diag}(1,1)$.
Next, leveraging the results given in Remark~\ref{safe_filter_tun}, we design a QP-synthesized controller, tunable universal formulas, and Sontag's formula under the safety filter framework. Specifically, we compute $c(\mathbf{x})= 1.5h(\mathbf{x})$, $\mathbf{d}(\mathbf{x})=[0,-1]$, and $\bar{c}(\mathbf{x})=c(\mathbf{x})+\mathbf{d}(\mathbf{x})\mathbf{k}_{0,\mathrm{d}}(\mathbf{x})$ (cf. Remark~\ref{safe_filter_tun}). Next, according to Theorem~\ref{Property_for_U_with_Kappa}, we set $\eta=0.5, 0.6, 0.7, 0.8, 0.9$ and obtain $\kappa_{1}(\mathbf{x}), \kappa_{2}(\mathbf{x}), \kappa_{3}(\mathbf{x}), \kappa_{4}(\mathbf{x}), \kappa_{5}(\mathbf{x})\in\mathcal{K}_{\mathrm{SM}}$ according to $\kappa(\mathbf{x})=(1-\eta)\cdot\frac{\bar{c}(\mathbf{x})}{\bar{\Gamma}_{\mathrm{Stg}}(\mathbf{x})}+\eta$, where $\bar{\Gamma}_{\mathrm{Stg}}(\mathbf{x})$ is defined by~\eqref{Gamma_Choice} with $\bar{c}(\mathbf{x})$, and $s(\|\mathbf{d}(\mathbf{x}) \|^{2})=\sigma\|\mathbf{d}(\mathbf{x}) \|^{2}$ with $\sigma=0.2$.
By using these selections of $\kappa(\mathbf{x})$ and employing~\eqref{SF_Tunable Universal_Formula}, several tunable universal formulas are derived.

\begin{figure}[tp]\label{Comparison}
 \centering
    \makebox[0pt]{%
    \includegraphics[width=0.775\columnwidth]{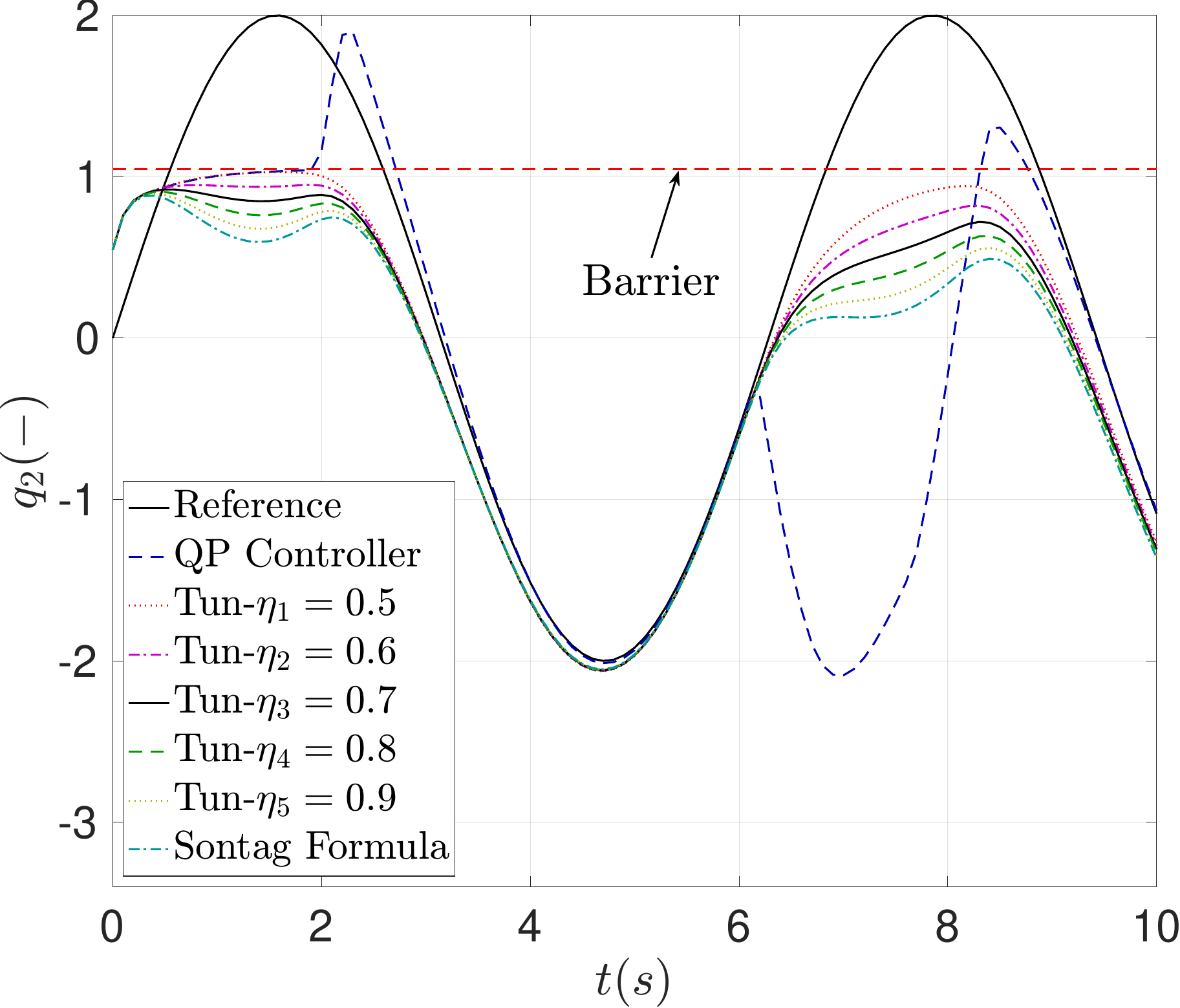}}
    \caption{Safe backstepping control using tunable universal formulas with different tunable terms, a QP-synthesized controller, and Sontag’s formula as virtual controllers.}
    \label{Safe_Backstepping}
\end{figure}

As shown in Fig.~\ref{Trajectory}, the safety-critical control behavior is depicted for different controllers resulting from various
$\kappa(\mathbf{x})$, each offering different degrees of robustness (or conservatism) while ensuring safety guarantees. Note that Sontag's formula exhibits conservative behavior in safety-critical control. Moreover, according to Proposition~\ref{Perfect_Approximation} (with $\eta = 0.5$ and $\sigma = 0.2$), it shows that the tunable universal formula closely approximates the QP-synthesized controller. In Fig.~\ref{Input_Behavior}, the control inputs $-\mathbf{v}_{2} =  \mathbf{k}_{\mathrm{d}}(\mathbf{x})-\mathbf{k}_{\mathrm{SF}}(\mathbf{x})$  from our tunable universal formulas (with a sign inversion), are smooth. This smoothness follows from Theorem~\ref{Property_for_U_with_Kappa}, due to the smooth functions $\eta_{1}(\mathbf{x}), \eta_{2}(\mathbf{x}), \eta_{3}(\mathbf{x}), \eta_{4}(\mathbf{x})$, $\eta_{5}(\mathbf{x})\in\Xi$. However, as indicated by the red circles in Fig.~\ref{Input_Behavior}, the QP-synthesized controller exhibits non-smooth behavior at certain points, which is consistent with our expectations. Once the virtual controller $\mathbf{v} = \mathbf{k}_0(\mathbf{x})$ is obtained, we construct a new barrier function $\mathfrak{b}(\mathbf{q},\mathbf{v}) = h(\mathbf{x}) - \frac{1}{2\mu}\left(\mathbf{v} - \mathbf{k}_0(\mathbf{x})\right)^\top\left(\mathbf{v} - \mathbf{k}_0(\mathbf{x})\right)$ with $\mu=20$. Following the approach given in~\cite{taylor2022safe}, we use a QP-synthesized safety filter (with CBF $\mathfrak{b}(\mathbf{q}, \mathbf{v})$) to design a control input $\mathbf{u}$ that ensures the safety of the system~\eqref{Segway_reformulation}. Note that the safety filter requires a nominal controller, and we define $\mathbf{k}_{\mathrm{d}}(\mathbf{x}) = \mathbf{H}^{-1}(-\Phi+\dot{\mathbf{k}}_{0}-\mathbf{K}_{\bar{\mathrm{P}}}(\mathbf{v}-\mathbf{k}_{0}))$ with $\mathbf{K}_{\bar{\mathrm{P}}}=\mathrm{diag}(1,1)$. As illustrated in Fig.~\ref{Safe_Backstepping}, the QP-synthesized controller fails to achieve safe tracking due to its non-smooth behavior.

Additionally, we study the effectiveness of tunable universal formulas in addressing safety-critical control when a norm-bounded control input constraint is considered. Specifically, we consider the virtual input $\mathbf{v}_{2}$ to be constrained by a norm-bounded constraint, such that $\|\mathbf{v}_{2}\| \leq \gamma$, where $\gamma = 2.3$ (corresponding to the red dashed line in Fig.~\ref{Input_Behavior}). It is observed in Fig.~\ref{Input_Behavior} that $\eta=0.8,0.9$ and Sontag's formula are not valid choices anymore since their controls exceed the norm bound $\gamma=2.3$. In this case, we should switch to the tunable universal formula provided in~\eqref{QP_CLF_Solution_with_tuning} (note that it also has to be adapted to a safety filter framework). 
Specifically, based on Theorem~\ref{Smoothness}, we find that only the controllers with $\eta = 0.5$, $0.6$, and $0.7$ (as shown in Fig.~\ref{Input_Behavior}) satisfy the safety and smoothness properties established in Theorem~\ref{Scaling_Universal_Formula} and Theorem~\ref{Smoothness}. This finding is consistent with the conclusions discussed in Section~\ref{Norm_Discussions}.

\section{Conclusions}\label{Conclusions}
This paper provides a novel framework for deriving a tunable universal formula for safety-critical control. By introducing a tunable term into Sontag's formula, we establish necessary and sufficient conditions for balancing safety, smoothness, and robustness in a tunable universal formula. In contrast to the approach in~\cite{cohen2023characterizing}, our approach is simpler, requiring only that the tunable term meets specific conditions and validity ranges, which enhances interpretability by clearly illustrating its effects on safety and robustness. Moreover, we emphasize that this tunable universal formula provides additional benefits for safety-critical tasks with norm-bounded input constraints; a straightforward adjustment to the tunable term's validity range enables the controller to regulate performance in safety, smoothness, and robustness while complying with these constraints. The effectiveness and advantages of the tunable universal formula are demonstrated through a two-link manipulator application, which highlights its capability to regulate desirable properties.

\bibliographystyle{IEEEtran}
\bibliography{Safety_Unify.bib}
\end{document}